\documentclass[preprint,11pt]{elsarticle}
\usepackage{amsmath,amsfonts,amssymb,amsthm}
\usepackage{mathtools}
\usepackage{graphicx}
\usepackage{color}
\usepackage{amsmath}
\usepackage{comment}
\usepackage{amssymb}
\usepackage{url}
 \usepackage{lineno}

\newtheorem{lemma}{Lemma}
\newtheorem{theorem}{Theorem}

\newtheorem{corollary}{Corollary}
\newtheorem{definition}{Definition}
\newtheorem{remark}{Remark}
\newtheorem{claim}{Claim}

\usepackage{todonotes}
\newcommand{\mmfvs}{\ensuremath\mathrm{mmfvs}}

\journal{Journal of Computer and System Sciences}

\begin{document}

\begin{frontmatter}

\title{(In)approximability of Maximum Minimal FVS\tnoteref{preversion}}

\tnotetext[preversion]{A
preliminary version of the paper appeared in Proceedings of the 31st International Symposium on Algorithms and Computation (ISAAC 2020), Vol.~181 of Leibniz International Proceedings in Informatics (LIPIcs),
Schloss Dagstuhl--Leibniz-Zentrum f{\"u}r Informatik, pp. 3:1--3:14 (2020) \cite{Confpaper}.}

\author[LAMSADE]{Louis Dublois}
\ead{louis.dublois@gmail.com}

\author[chuo]{Tesshu Hanaka}
\ead{hanaka.91t@g.chuo-u.ac.jp}

\author[LAMSADE]{Mehdi Khosravian Ghadikolaei}
\ead{m.khosravian@gmail.com}

\author[LAMSADE]{Michael Lampis}
\ead{michail.lampis@dauphine.fr}

\author[LAMSADE]{Nikolaos Melissinos}
\ead{nikolaos.melissinos@dauphine.eu}

\address[LAMSADE]{Universit\'e Paris-Dauphine, PSL Research University, CNRS, UMR 7243, LAMSADE, Paris, France}
\address[chuo]{Chuo University, Tokyo, Japan}

\begin{abstract}
We study the approximability of the NP-complete \textsc{Maximum Minimal
Feedback Vertex Set} problem. Informally, this natural problem seems to lie in
an intermediate space between two more well-studied problems of this type:
\textsc{Maximum Minimal Vertex Cover}, for which the best achievable
approximation ratio is $\sqrt{n}$, and \textsc{Upper Dominating Set}, which
does not admit any $n^{1-\epsilon}$ approximation. We confirm and quantify this
intuition by showing the first non-trivial polynomial time approximation for
\textsc{Max Min FVS} with a ratio of $O(n^{2/3})$, as well as a matching hardness
of approximation bound of $n^{2/3-\epsilon}$, improving the previous known
hardness of $n^{1/2-\epsilon}$. 
The approximation algorithm also gives a cubic kernel when parameterized by the solution size.
Along the way, we also obtain an
$O(\Delta)$-approximation and show that this is asymptotically best possible,
and we improve the bound for which the problem is NP-hard from $\Delta\ge 9$ to
$\Delta\ge 6$. 

Having settled the problem's approximability in polynomial time, we move to the
context of super-polynomial time. We devise a generalization of our
approximation algorithm which, for any desired approximation ratio $r$,
produces an $r$-approximate solution in time $n^{O(n/r^{3/2})}$. This
time-approximation trade-off is essentially tight: we show that under the ETH,
for any ratio $r$ and $\epsilon>0$, no algorithm can $r$-approximate this
problem in time $n^{O((n/r^{3/2})^{1-\epsilon})}$, hence we precisely
characterize the approximability of the problem for the whole spectrum between
polynomial and sub-exponential time, up to an arbitrarily small constant in the
second exponent. 
\end{abstract}


\end{frontmatter}


\section{Introduction}\label{sec:intro}

In a graph $G=(V,E)$, a set $S\subseteq V$ is called a \emph{feedback vertex
set} (fvs for short) if the subgraph induced by $V\setminus S$ is a forest.
Typically, fvs is studied with a minimization objective: given a graph we are
interested in finding the best (that is, smallest) fvs.  In this paper we are
interested in an objective which is, in a sense, the inverse: we seek an fvs
$S$ which is as \emph{large} as possible, while still being minimal. We call
this problem \textsc{Max Min FVS}.

MaxMin and MinMax versions of many famous optimization problems have recently
attracted much interest in the literature (we give references below) and
\textsc{Max Min FVS} can be seen as a member of this framework.  Although the
initial motivation for studying such problems was a desire to analyze the worst
possible performance of a naive heuristic, these problems have gradually been
revealed to possess a rich combinatorial structure that makes them interesting
in their own right. Our goal in this paper is to show that \textsc{Max Min FVS}
displays an interesting complexity behavior with respect to its
approximability.

Our motivation for focusing on \textsc{Max Min FVS} is the contrast between two
of its more well-studied cousins: the \textsc{Max Min Vertex Cover} and
\textsc{Upper Dominating Set} problems, where the objective is to find the
largest minimal vertex cover or dominating set respectively. At first glance,
one would expect \textsc{Max Min VC} to be the easier of these two problems:
both problems can be seen as trying to find the largest minimal hitting set of
a hypergraph, but in the case of \textsc{Max Min VC} the hypergraph has a very
restricted structure, while in \textsc{UDS} the hypergraph is essentially
arbitrary. This intuition turns out to be correct: while \textsc{UDS} admits no
$n^{1-\epsilon}$-approximation \cite{Bazgan2018}, \textsc{Max Min VC} admits a
$\sqrt{n}$-approximation (but no $n^{1/2-\epsilon}$-approximation)
\cite{Boria2015}.

This background leads us to the natural question of the approximability of
\textsc{Max Min FVS}. On an intuitive level, one may be tempted to think that
this problem should be harder than \textsc{Max Min VC}, since hitting cycles is
more complex than hitting edges, but easier than \textsc{UDS}, since hitting
cycles still offers us more structure than an arbitrary hypergraph. However, to
the best of our knowledge, no $n^{1-\epsilon}$-approximation algorithm is
currently known for \textsc{Max Min FVS} (so the problem could be as hard as
\textsc{UDS}), and the best hardness of approximation bound known is
$n^{1/2-\epsilon}$ \cite{MishraS01} (so the problem could be as easy as
\textsc{Max Min VC}).

Our main contribution in this paper is to fully answer this question, confirming and
precisely quantifying the intuition that \textsc{Max Min FVS} is a
problem that lies ``between'' \textsc{Max Min VC} and \textsc{UDS}: We give a
polynomial-time approximation algorithm with ratio $O(n^{2/3})$ and a hardness
of approximation reduction which shows that (unless $\text{P}=\text{NP}$) no polynomial-time
algorithm can obtain a ratio of $n^{2/3-\epsilon}$, for any $\epsilon>0$. This
completely settles the approximability of the problem in polynomial time. 
Along the way, we also prove that  \textsc{Max Min FVS} admits a cubic kernel when parameterized by the solution size, give an approximation algorithm with ratio $O(\Delta)$, show
that no algorithm can achieve ratio $\Delta^{1-\epsilon}$, for any
$\epsilon>0$, and improve the best known NP-completeness proof for \textsc{Max
Min FVS} from $\Delta\ge 9$ \cite{MishraS01} to $\Delta \ge 6$, where $\Delta$
is the maximum degree of the input graph.

One interesting aspect of our results is that they have an interpretation from
extremal combinatorics which nicely mirrors the situation for \textsc{Max Min
VC}. Recall that a corollary of the $\sqrt{n}$-approximation for \textsc{Max
Min VC} \cite{Boria2015} is that any graph without isolated vertices has a
minimal vertex cover of size at least $\sqrt{n}$, and this is tight (see Remark
\ref{rem:vc}). Hence, the algorithm only needs to trivially preprocess the
graph (deleting isolated vertices) and then find this set, which is guaranteed
to exist.  Our algorithms can be seen in a similar light: we prove that if one
applies two almost trivial pre-processing rules to a graph (deleting leaves and
contracting edges between degree-two vertices), a minimal fvs of size at least
$n^{1/3}$ (and $\Omega(n/\Delta)$) is always guaranteed to exist, and this is
tight (Corollary \ref{cor:extr} and Remark \ref{rem:fvs}).  Thus, the
approximation ratio of $n^{2/3}$ is automatically guaranteed for any graph
where we exhaustively apply these very simple rules and our algorithms only
have to work to construct the promised set. This makes it somewhat remarkable
that the ratio of $n^{2/3}$ turns out to be best possible.

Having settled the approximability of \textsc{Max Min FVS} in polynomial time,
we consider the question of how much time needs to be invested if one wishes to
guarantee an approximation ratio of $r$ (which may depend on $n$) where $r
<n^{2/3}$. This type of time-approximation trade-off was extensively studied by
Bonnet et al. \cite{Bonnet2018}, who showed that \textsc{Max Min Vertex Cover}
admits an $r$-approximation in time $2^{O(n/r^2)}$ and this is optimal under
the randomized ETH.

For \textsc{Max Min FVS} we cannot hope to obtain a trade-off with performance
exponential in $n/r^2$, as this implies a polynomial-time
$\sqrt{n}$-approximation. It therefore seems more natural to aim for a running
time exponential in $n/r^{3/2}$. Indeed, generalizing our polynomial-time
approximation algorithm, we show that we can achieve an $r$-approximation in
time $n^{O(n/r^{3/2})}$. Although this algorithm reuses some ingredients from
our polynomial-time approximation, it is significantly more involved, as it is
no longer sufficient to compare the size of our solution to $n$. We complement
our result with a lower bound showing that our algorithm is essentially best
possible under the randomized ETH for any $r$ (not just for polynomial time),
or more precisely that the exponent of the running time of our algorithm can
only be improved by $n^{o(1)}$ factors.

\bigskip
 \noindent{\bf Related work} To the best of our knowledge, \textsc{Max Min FVS}
was first considered by Mishra and Sikdar \cite{MishraS01}, who showed that the problem does not admit an $n^{1/2-\epsilon}$ approximation (unless $\text{P}=\text{NP}$), and that it remains APX-hard for $\Delta\ge 9$. On the other hand, \textsc{UDS} and
\textsc{Max Min VC} are well-studied problems, both in the context of
approximation and in the context of parameterized complexity
\cite{AbouEishaHLMRZ18,Bazgan2018,Boria2015,BoyaciM17,ChestonFHJ90,CourcelleMR00,Demange1999,HenningP20,JacobsonP90,Zehavi2017,ZZ1995}.
Many other classical optimization problems have recently been studied in the MaxMin or MinMax framework, such as \textsc{Max Min Separator}
\cite{Hanaka2019}, \textsc{Max Min Cut} \cite{EtoHKK2019}, \textsc{Min Max
Knapsack} (also known as the \textsc{Lazy Bureaucrat Problem})
\cite{ArkinBMS03,Furini2017,GourvesMP13}, and some variants of \textsc{Max Min Edge Cover} \cite{KhoshkhahGMS20, DBLP:conf/ttcs/HarutyunyanGMMP20}. Some problems in this area also arise naturally in other
forms and have been extensively studied, such as \textsc{Min Max Matching}
(also known as \textsc{Edge Dominating Set} \cite{IwaideN16}),
\textsc{Grundy Coloring}, which can be seen as a Max Min version of
\textsc{Coloring} \cite{AboulkerB0S20,BelmonteKLMO20}, and \textsc{Max Min VC} in hypergraphs, which is known as \textsc{Upper Transversal}\cite{DBLP:conf/wg/MisraRRS12,
DBLP:journals/combinatorics/HenningY18, DBLP:journals/ejc/HenningY19, DBLP:journals/jco/HenningY20}. 

The idea of designing super-polynomial time approximation algorithms which
obtain guarantees better than those possible in polynomial time has attracted
much attention in the last decade
\cite{BansalCLNN19,BourgeoisEP09,CyganKW09,CyganP10,EscoffierPT14,FotakisLP16,KatsikarelisLP19}.
As mentioned, the result closest to the time-approximation trade-off we give in
this paper is the approximation algorithm for \textsc{Max Min VC} given by
Bonnet et al. \cite{Bonnet2018}. It is important to note that such trade-offs
are only generally known to be tight up to poly-logarithmic factors in the
exponent of the running time. As explained in \cite{Bonnet2018}, current lower
bound techniques can rule out improvements in the running time that shave at
least $n^\epsilon$ from the exponent, but not improvements which shave
poly-logarithmic factors, due to the state of the art in quasi-linear PCP
constructions. Indeed, such improvements are sometimes possible
\cite{BansalCLNN19} and are conceivable for \textsc{Max Min VC} and \textsc{Max
Min FVS}. Lower bounds for this type of algorithm rely on the (randomized)
Exponential Time Hypothesis (ETH), which states that there is no (randomized)
algorithm for \textsc{3-SAT} running in time $2^{o(n)}$.

\section{Preliminaries}

We use standard graph-theoretic notation and only consider simple loop-less
graphs.  For a graph $G=(V,E)$ and $S\subseteq V$ we denote by $G[S]$ the graph
induced by $S$. For $u\in V$, $G-u$ is the graph $G[V\setminus\{u\}]$. We write
$N(u)$ to denote the set of neighbors of $u$ and $d(u)=|N(u)|$ to denote its
degree. For $S\subseteq V$, $N(S)=\cup_{u\in S}N(u)\setminus S$. We use
$\Delta(G)$ (or simply $\Delta$) to denote the maximum degree of $G$. For
$uv\in E$ the graph $G/uv$ is the graph obtained by contracting the edge $uv$,
that is, replacing $u,v$ by a new vertex connected to $N(u)\cup N(v)$. In this
paper we will only apply this operation when $N(u)\cap N(v)=\emptyset$, so the
result will always be a simple graph.  

A forest is a graph that does not contain cycles. A feedback vertex set (fvs
for short) is a set $S\subseteq V$ such that $G[V\setminus S]$ is a forest. An
fvs $S$ is minimal if no proper subset of $S$ is an fvs. It is not hard to see
that if $S$ is minimal, then every $u\in S$ has a \emph{private cycle}, that
is, there exists a cycle in $G[(V\setminus S)\cup \{u\}]$, which goes through
$u$.  A vertex $u$ of a feedback vertex set $S$ that does not have a private
cycle (that is, $S\setminus\{u\}$ is also an fvs), is called \emph{redundant}.
For a given fvs $S$, we call the set $F:=V\setminus S$ the corresponding
induced forest. If $S$ is minimal, then $F$ is maximal.

The main problem we are interested in is \textsc{Max Min FVS}: given a
graph $G=(V,E)$, find a minimal fvs of $G$ of maximum size.  Since this problem
is NP-hard, we will be interested in approximation algorithms. An approximation
algorithm with ratio $r\ge 1$ (which may depend on $n$, the order of the graph)
is an algorithm which, given a graph $G$, returns a solution of size at least
$\frac{\mmfvs(G)}{r}$, where $\mmfvs(G)$ is the size of the largest minimal fvs
of $G$.

We make two basic observations about our problem: deleting vertices or
contracting edges can only decrease the size of the optimal solution.

\begin{lemma}\label{lem:ur1} Let $G=(V,E)$ be a graph and $u\in V$. Then,
$\mmfvs(G)\ge \mmfvs(G-u)$. Furthermore, given any minimal feedback vertex set
$S$ of $G-u$, it is possible to construct in polynomial time a minimal feedback
vertex set of $G$ of the same or larger size. \end{lemma}

\begin{proof} Let $S$ be a minimal fvs of $G-u$. We
observe that $S\cup\{u\}$ is an fvs of $G$. If $S\cup\{u\}$ is minimal, we are
done. If not, we delete vertices from it until it becomes minimal. We now note
that the only vertex which may be deleted in this process is $u$, since all
vertices of $S$ have a private cycle in $G-u$ (that is, a cycle not intersected
by any other vertex of $S$). Hence, the resulting set is a superset of $S$.
\end{proof}

\begin{lemma}\label{lem:ur2} Let $G=(V,E)$ be a graph, $u,v\in V$ with
$N(u)\cap N(v)=\emptyset$ and $uv\in E$.  Then $\mmfvs(G) \ge \mmfvs(G/uv) $.
Furthermore, given any minimal feedback vertex set $S$ of $G/uv$, it is
possible to construct in polynomial time a minimal feedback vertex set of $G$
of the same or larger size.  \end{lemma}

\begin{proof} Before we prove the Lemma we note that the contraction operation,
under the condition that $N(u)\cap N(v)=\emptyset$, preserves acyclicity in a
strong sense: $G$ is acyclic if and only if $G/uv$ is acyclic. Indeed, if we
contract an edge that is part of a cycle, this cycle must have length at least
$4$, and will therefore give a cycle in $G/uv$. Of course, contractions never
create cycles in acyclic graphs.

Let $G':=G/uv$, $w$ be the vertex of $G'$ which has replaced $u,v$, $V'=V(G')$,
and $S$ be a minimal fvs of $G'$. We have two cases: $w\in S$ or $w\not\in S$.

In case $w\in S$, we start with the set $S'=(S\setminus\{w\})\cup\{u,v\}$. It
is not hard to see that $S'$ is an fvs of $G$. Furthermore, no vertex of
$S'\setminus\{u,v\}$ is redundant: for all $z\in S\setminus\{w\}$, there is a
cycle in $G'[(V'\setminus S) \cup \{z\}]$, therefore there is also a cycle in
$G[(V\setminus S')\cup \{z\}]$. Furthermore, we claim that $S'\setminus\{u,v\}$
is not a valid fvs. Indeed, there must be a cycle contained (due to minimality)
in $G_1=G'[(V'\setminus S) \cup\{w\}]$. Therefore, if there is no cycle in
$G_2=G[(V\setminus S')\cup\{u,v\}]$, we get a contradiction, as $G_1$ can be
obtained by $G_2$ by contracting the edge $uv$ and contracting edges preserves
acyclicity. We conclude that even if $S'$ is not minimal, if we remove vertices
until it becomes minimal, we will remove at most one vertex, so the size of the
fvs obtained is at least $|S|$.

In case $w\not\in S$, we will return the same set $S$. Let $F=V\setminus S,
F'=V'\setminus S$.  By definition, $G'[F']$ is acyclic. To see that $G[F]$ is
also a forest, we note that $G'[F']$ is obtained from $G[F]$ by contracting
$uv$, and as we noted in the beginning, the contractions we use strongly
preserve acyclicity.  To see that $S$ is minimal, take $z\in S$ and consider
the graphs $G_1=G[(V\setminus S)\cup \{z\}]$ and $G_2=G'[(V'\setminus
S)\cup\{z\}]$.  We see that $G_2$ can be obtained from $G_1$ by contracting
$uv$. But $G_2$ must have a cycle, by the minimality of $S$, so $G_1$ also has
a cycle. Thus, $S'$ is minimal in $G$.  \end{proof}

\section{Polynomial Time Approximation Algorithm}

In this section we present a polynomial-time algorithm which guarantees an
approximation ratio of $n^{2/3}$. As we show in Theorem \ref{thm:Inapprox},
this ratio is the best that can be hoped for in polynomial time. Later (Theorem
\ref{thm:subexp}) we show how to generalize the ideas presented here to obtain
an algorithm that achieves a trade-off between the approximation ratio and the
(sub-exponential) running time, and show that this trade-off is essentially
optimal.

On a high level, our algorithm proceeds as follows: first we identify some easy
cases in which applying Lemma \ref{lem:ur1} or Lemma \ref{lem:ur2} is
\emph{safe}, that is, the value of the optimal is guaranteed to stay constant,
namely deleting vertices of degree at most $1$, and contracting edges between
vertices of degree $2$. After we apply these reduction rules exhaustively, we
compute a minimal fvs $S$ in an arbitrary way. If $S$ is large enough (larger
than $n^{1/3}$), we simply return this set.

If not, we apply some counting arguments to show that a vertex $u\in S$ with
high degree ($\ge n^{2/3}$) must exist. We then have two cases: either we are
able to construct a large minimal fvs just by looking at the neighborhood of
$u$ in the forest (and ignoring $S\setminus\{u\}$), or $u$ must share many
neighbors with another vertex $v\in S$, in which case we construct a large
minimal fvs in the common neighborhood of $u,v$.

Because our algorithm is constructive (and runs in polynomial time), we find it
interesting to remark an interpretation from the point of view of extremal
combinatorics, given in Corollary \ref{cor:extr}.

\subsection{Basic Reduction Rules and Combinatorial Tools}

We begin by showing two \emph{safe} versions of Lemmas \ref{lem:ur1},
\ref{lem:ur2}.

\begin{lemma}\label{lem:sr1} Let $G,u$ be as in Lemma \ref{lem:ur1} with
$d(u)\le 1$. Then $\mmfvs(G-u)=\mmfvs(G)$. \end{lemma}

\begin{proof} We only need to show that
$\mmfvs(G)\le\mmfvs(G-u)$ (the other direction is given by Lemma
\ref{lem:ur1}). Let $S$ be a minimal fvs of $G$.  Then, $S$ is an fvs of $G-u$.
Furthermore, $u\not\in S$, as $S$ is minimal in $G$. To see that $S$ is also
minimal in $G-u$, note that any cycle of $G$ also exists in $G-u$ (as no cycle
contains $u$).  \end{proof}

\begin{lemma}\label{lem:sr2} Let $G,u,v$ be as in Lemma \ref{lem:ur2} with
$d(u)=d(v)=2$. Then $\mmfvs(G/uv)=\mmfvs(G)$. \end{lemma}

\begin{proof} Let $G'=G/uv$, $w$ be the vertex that replaced $u,v$ in $G'$, and
$V'=V(G')$.

We only need to show that $\mmfvs(G)\le\mmfvs(G')$, as the other direction is
given by Lemma \ref{lem:ur2}. Let $S$ be a minimal fvs of $G$. We consider two
cases:

If $u,v\not\in S$, then we claim that $S$ is also a minimal fvs of $G'$.
Indeed, $G'[V'\setminus S]$ is obtained from $G[V\setminus S]$ by contracting
$uv$, so both are acyclic. Furthermore, for all $z\in S$, $G'[(V'\setminus S)
\cup\{z\}]$ is obtained from $G[(V\setminus S)\cup\{z\}]$ by contracting $uv$,
therefore both have a cycle, hence no vertex of $S$ is redundant in $G'$.

If $\{u,v\}\cap S\neq \emptyset$, we claim that exactly one of $u,v$ is in $S$.
Indeed, if $u,v\in S$, then $G[(V\setminus S)\cup \{u\}]$ does not contain a
cycle going through $u$, as $u$ has degree $1$ in this graph. Without loss of
generality, let $u\in S$, $v\not\in S$. We set $S':=(S\setminus\{u\})\cup\{w\}$
and claim that $S'$ is a minimal fvs of $G'$. It is not hard to see that $S'$
is an fvs of $G'$, since it corresponds to deleting $S\cup\{v\}$ from $G$. To
see that it is minimal, for all $z\in S'\setminus\{w\}$ we observe that
$G'[(V'\setminus S')\cup\{z\}]$ obtained from $G'[(V\setminus S)\cup\{z\}]$ by
deleting $v$, which has degree $1$. Therefore, this deletion strongly preserves
acyclicity. Finally, to see that $w$ is not redundant for $S'$ we observe that
$G[(V\setminus S)\cup\{u\}]$ has a cycle, and this cycle must be present in
$G'[(V'\setminus S')\cup\{w\}]$, which is obtained from the former graph by
contracting $uv$.  \end{proof}

\begin{definition} For a graph $G=(V,E)$ we say that $G$ is reduced if it is
not possible to apply Lemma \ref{lem:sr1} or Lemma \ref{lem:sr2} to $G$.
\end{definition}

We now present a counting argument which will useful in our algorithm and
states, roughly, that if in a reduced graph we find a (not necessarily minimal)
fvs, that fvs must have many neighbors in the corresponding forest.

\begin{lemma}\label{lem:manyneighbors} Let $G=(V,E)$ be a reduced graph and
$S\subseteq V$ a feedback vertex set of $G$. Let $F=V\setminus S$. Then,
$|N(S)\cap F|\ge \frac{|F|}{4}$.  \end{lemma}

\begin{proof} Let $n_1$ be the number of leaves
of $F$, $n_3$ the number of vertices of $F$ with at least three neighbors in
$F$, $n_{2a}$ the number of vertices of $F$ with two neighbors in $F$ and at
least one neighbor in $S$, and $n_{2b}$ the number of remaining vertices of
$F$. We have $n_1+n_{2a}+n_{2b}+n_3=|F|$.  Furthermore,  $n_3\le n_1$ because
the average degree of any forest is less than $2$.

We observe that all leaves of the tree have a neighbor in $S$ (otherwise we
would have applied Lemma \ref{lem:sr1}). This gives $|N(S)\cap F|\ge
n_1+n_{2a}$. 

Furthermore, none of the $n_{2b}$ vertices which have degree two in the tree
and no neighbors in $S$ can be connected to each other, since then Lemma
\ref{lem:sr2} would apply. Therefore, $n_{2b}\le n_1+n_{2a}+n_3$. Indeed, if
$n_{2b}>n_1+n_{2a}+n_3$, then $n_{2b}>|F|/2$ and, since these $n_{2b}$ vertices
form an independent set, we would have $|E(F)|\ge 2n_{2b} >|F|$, contradicting
the assumption that $F$ is a forest.

Putting things together we get $|F|=n_1+n_{2a}+n_{2b}+n_3 \le 2n_1+2n_{2a}
+2n_3 \le 4n_1+2n_{2a} \le 4|N(S)\cap F|$.  \end{proof}

We note that Lemma \ref{lem:manyneighbors} immediately gives an approximation
algorithm with ratio $O(\Delta)$.

\begin{lemma}\label{lem:delta} In a reduced graph $G$ with $n$ vertices and
maximum degree $\Delta$, every feedback vertex set has size at least
$\frac{n}{5\Delta}$.  \end{lemma}

\begin{proof} Let $S$ be a feedback vertex set of $G$
and $F$ the corresponding forest. If $|S|<\frac{n}{5\Delta}$ then $|N(S)\cap
F|< \frac{n}{5}$ so by Lemma \ref{lem:manyneighbors} we have
$|F|<\frac{4n}{5}$.  But then $|V|=|S|+|F|<n$, which is a contradiction.
\end{proof}

\begin{remark}\label{rem:t1}Lemma \ref{lem:manyneighbors} is tight.\end{remark}

\begin{proof} Take two copies of a rooted binary tree with
$n$ leaves and connect their roots. The resulting tree has $2n$ leaves and
$2n-2$ vertices of degree $3$.  Subdivide every edge of this tree. Add two
vertices $u,v$ connected to every leaf. In the resulting graph $S=\{u,v\}$ is
an fvs.  The corresponding forest has $8n-5$ vertices, of which $2n$ are
connected to $S$.  The graph is reduced. \end{proof}

\subsection{Polynomial Time Approximation and Extremal Results}

We begin with a final intermediate lemma that allows us to construct a large
minimal fvs in any reduced graph that is a forest plus one vertex. 

\begin{lemma}\label{lem:onevertex} Let $G=(V,E)$ be a reduced graph and $u\in
V$ such that $G-u$ is acyclic. Then it is possible to construct in polynomial
time a minimal feedback vertex set $S$ of $G$ with $|S|\ge d(u)/2$. \end{lemma}

\begin{proof} Let $F=V\setminus\{u\}$. Since the
graph is reduced, all trees of $G[F]$ contain at least two neighbors of $u$.
Indeed, since each tree $T$ of $G[F]$ has at least two leaves, both of them
must be neighbors of $u$ (otherwise Lemma \ref{lem:sr1} applies).

As long as there exist $v,w\in F$ with $vw\in E$ and $\{v,w\}\not\subseteq
N(u)$ we contract the edge $vw$.  Note that we can apply Lemma \ref{lem:ur2} as
$v,w$ have no common neighbors ($u$ is not a common neighbor by assumption, and
they cannot have a common neighbor in the forest without forming a cycle).
Furthermore, this operation does not change $d(u)$.  Therefore, it will be
sufficient to construct a minimal fvs in the resulting graph after applying
this operation exhaustively.

Suppose now that we have applied the contraction operation described above
exhaustively.  We eventually arrive at a graph where $u$ is connected to all
vertices of $F$, as all trees of $F$ initially contain some neighbors of $u$
and, after repeated contractions, all non-neighbors of $u$ are absorbed into
its neighbors (more precisely, each contraction decreases $|F\setminus N(u)|$).
Therefore, we arrive at a graph with $d(u)=|F|$.  Furthermore, every component
of $F$ contains strictly more than one vertex.

Now, since $G[F]$ is bipartite, there is a bipartition $F=L\cup R$. Without
loss of generality $|L|\le |R|$. We return the solution $S=R$. First, $S$ does
have the promised size, as $|S|\ge |F|/2 = d(u)/2$. Second, $S$ is an fvs, as
$L$ is an independent set, so $L\cup\{u\}$ induces a star. Finally, $S$ is
minimal, because all $v\in S$ are connected to $u$, and also have at least one
neighbor  $w\in L$, with $w$ also connected to $u$. An illustration of the process is presented in Figure \ref{Fig:Exm_mmfvs_tree}.

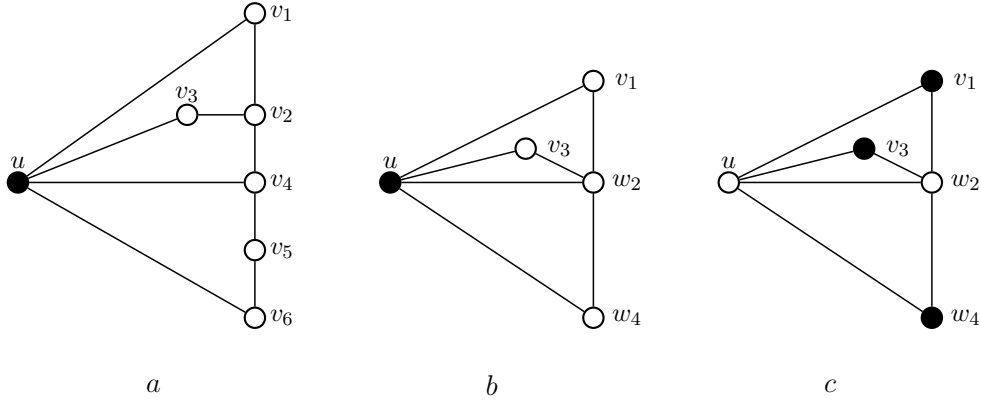
\begin{figure}
\centering
\begin{tikzpicture}[scale=0.9, transform shape]
\tikzstyle{vertex}=[circle, draw, inner sep=2pt,  minimum width=2 pt, minimum size=0.3cm, thick]

\node[vertex, fill = black] (u) at (-0.5,0) {};
\node[vertex] (v1) at (3,2.5) {};
\node[vertex] (v2) at (3,1) {};
\node[vertex] (v3) at (2,1) {};
\node[vertex] (v4) at (3,0) {};
\node[vertex] (v5) at (3,-1) {};
\node[vertex] (v6) at (3,-2) {};

\draw (u)[line width = 0.2 mm]--(v1)--(v2)--(v4)--(v5)--(v6)--(u);
\draw (v2)[line width = 0.2 mm]--(v3)--(u)--(v4);
\node () at (-0.5,0.3) {$u$};
\node () at (3.4,2.5) {$v_1$};
\node () at (3.4,1) {$v_2$};
\node () at (2,1.3) {$v_3$};
\node () at (3.4,0) {$v_4$};
\node () at (3.4,-1) {$v_5$};
\node () at (3.4,-2) {$v_6$};
\node () at (1.5,-3) {\large{$a$}};

\begin{scope}[xshift=5cm]

\node[vertex, fill = black] (u) at (0,0) {};
\node[vertex] (w1) at (3,1.5) {};
\node[vertex] (w3) at (2,0.5) {};
\node[vertex] (w2) at (3,0) {};
\node[vertex] (w4) at (3,-2) {};

\draw (u)[line width = 0.2 mm]--(w1)--(w2)--(u)--(w3)--(w2)--(w4)--(u);
\node () at (0,0.3) {$u$};
\node () at (3.5,1.5) {$v_1$};
\node () at (3.5,0) {$w_2$};
\node () at (2.5,0.5) {$v_3$};
\node () at (3.5,-2) {$w_4$};

\node () at (1.5,-3) {\large{$b$}};
\end{scope}

\begin{scope}[xshift=10cm]
\node[vertex] (u) at (0,0) {};
\node[vertex, fill = black] (w1) at (3,1.5) {};
\node[vertex, fill = black] (w3) at (2,0.5) {};
\node[vertex] (w2) at (3,0) {};
\node[vertex, fill = black] (w4) at (3,-2) {};

\draw (u)[line width = 0.2 mm]--(w1)--(w2)--(u)--(w3)--(w2)--(w4)--(u);
\node () at (0,0.3) {$u$};
\node () at (3.5,1.5) {$v_1$};
\node () at (3.5,0) {$w_2$};
\node () at (2.5,0.5) {$v_3$};
\node () at (3.5,-2) {$w_4$};

\node () at (1.5,-3) {\large{$c$}};

\end{scope}

\end{tikzpicture}
\caption{$(a)$ vertex $u$ is a minimal fvs of the given graph and has 4 neighbors in $G[F]$. $(b)$ a contracted form of $G[F]$ with 4 vertices. $(c)$ a new minimal fvs of the result graph of size 3.}\label{Fig:Exm_mmfvs_tree}
\end{figure}

\end{proof}

\begin{theorem}\label{thm:algpoly} There is a polynomial time approximation
algorithm for \textsc{Max Min FVS} with ratio $O(n^{2/3})$.
\end{theorem}

\begin{proof} We are given a graph $G=(V,E)$. We begin by applying Lemmas
\ref{lem:sr1},\ref{lem:sr2} exhaustively in order to obtain a reduced graph
$G'=(V',E')$. Clearly, if we obtain a $|V'|^{1/3}$ approximation in $G'$, since
the reductions we applied do not change the optimal, and we can construct a
solution of the same size in $G$, we get a $|V'|^{2/3}\le |V|^{2/3}$
approximation ratio in $G$. So, in the remainder, to ease presentation, we
assume $G$ is already reduced and has $n$ vertices.

Our algorithm begins with an arbitrary minimal fvs $S$. This can be
constructed, for example, by starting with $S=V$ and removing vertices from $S$
until it becomes minimal. If $|S|\ge n^{1/3}$ then we return $S$. Since the
optimal solution cannot have size more than $n$, we achieve the claimed ratio.

Suppose then that $|S|<n^{1/3}$. Let $F$ be the corresponding forest. We have
$|F|>n-n^{1/3}>n/2$ for sufficiently large $n$. By Lemma
\ref{lem:manyneighbors}, $|N(S)\cap F|\ge n/8$. Since $|S|<n^{1/3}$ there must
exist $u\in S$ such that $u$ has at least $\frac{n^{2/3}}{8}$ neighbors in $F$.

Let $w\in F\cap N(u)$. We say that $w$ is a \emph{good} neighbor of $u$ if
there exists $w'\in F\cap N(u)$ with $w'\neq w$ and $w'$ is in the same tree of
$G[F]$ as $w$.  Otherwise $w$ is a bad neighbor of $u$. By extension, a tree of
$G[F]$ that contains a good (resp. bad) neighbor of $u$ will be called a good
(resp.  bad) tree. Every vertex of $N(u)\cap F$ is either good or bad.

We have argued that $|N(u)\cap F|\ge \frac{n^{2/3}}{8}$. We distinguish two
cases: either $u$ has at least $\frac{n^{2/3}}{16}$ good neighbors in $F$, or
it has at least that many bad neighbors in $F$.

In the former case, we delete from the graph the set $S\setminus \{u\}$ and apply
Lemmas \ref{lem:sr1}, \ref{lem:sr2} exhaustively again. We claim that the
number of good neighbors of $u$ does not decrease in this process.  Indeed, two
good neighbors of $u$ cannot be contracted using Lemma \ref{lem:sr2}, since
they have a common neighbor (namely $u$). Furthermore, suppose $w$ is the first
good neighbor of $u$ to be deleted using Lemma \ref{lem:sr1}. This would mean
that $w$ currently has no other neighbor except $u$. However, since $w$ is
good, initially there was a $w'\in N(u)$ in the same tree of $G[F]$ as $w$. The
vertex $w'$ has not been deleted (since we assumed $w$ is the first good
neighbor to be deleted). Furthermore, Lemmas \ref{lem:sr1}, \ref{lem:sr2}
cannot disconnect two vertices which are in the same component, so we get a
contradiction.  We therefore have a reduced graph, where $\{u\}$ is an fvs, and
$d(u)\ge \frac{n^{2/3}}{16}$. By Lemma \ref{lem:onevertex} we obtain a minimal
fvs of size at least $\frac{n^{2/3}}{32}$, which is an $O(n^{1/3})$
approximation.

In the latter case, $u$ has at least $\frac{n^{2/3}}{16}$ bad neighbors in $F$.
Consider a bad tree $T$. We claim that $T$ must have a neighbor in
$S\setminus\{u\}$, because $T$ has at least two leaves, at most one of which is
a neighbor of $u$ (since $T$ is bad). If the second leaf is not connected to
$S$, it will be deleted by Lemma \ref{lem:sr1}. Furthermore, since $u$ is
connected to one vertex in each bad tree, $u$ is connected to at least
$\frac{n^{2/3}}{16}$ bad trees.

We now find the vertex $v\in S\setminus\{u\}$ such that $v$ is connected to the
maximum number of bad trees connected to $u$. Since $|S|\le n^{1/3}$, $v$ must
be connected to at least $\frac{n^{1/3}}{16}$ bad trees connected to $u$. We
now delete from the graph the set $S\setminus\{u,v\}$ as well as all trees of
$G[F]$, except the bad trees connected to $u,v$.  Furthermore, in each bad tree
$T$ connected to both $u,v$ let $u'\in T\cap N(u)$ and $v'\in T\cap N(v)$ such
that $u',v'$ are as close as possible in $T$ (note that perhaps $v'=u'$).  We
delete all vertices of the tree $T$ except those on the path from $v'$ to $u'$.
Then, we contract all internal edges of this path (note that internal vertices
of the path are not connected to $\{u,v\}$ by the selection of $u',v'$).  It is
not hard to verify that, by using Lemmas \ref{lem:ur1}, \ref{lem:ur2}, if we
are able to produce a large minimal fvs in the resulting graph, we obtain a
solution for $G$.  Furthermore, in the resulting graph, every bad tree $T$
connected to $u,v$ has been reduced to a single vertex connected to $u,v$. So
the graph is now either a $K_{2,s}$, with $s\ge \frac{n^{1/3}}{16}$, or the
same graph with the addition of the edge $uv$. In either case, it is not hard to
see that starting with the fvs that contains all vertices except $\{u,v\}$, and
making it minimal, we obtain a solution of size at least $s-1$ which gives an
approximation ratio of $O(n^{2/3})$.  \end{proof}

\begin{corollary}\label{cor:extr} For any reduced graph $G$ on $n$ vertices we
have $\mmfvs(G)=\Omega( n^{1/3})$. \end{corollary}

\begin{proof} We simply note that the algorithm of
Theorem \ref{thm:algpoly} always constructs a solution of size at least
$\frac{n^{1/3}}{c}$, where $c$ is a small constant, assuming that the original
$n$-vertex graph $G$ was reduced.  \end{proof}

\begin{remark}\label{rem:fvs}Corollary \ref{cor:extr} is tight.\end{remark}

\begin{proof} Take a $K_n$ and for every pair of vertices
$u,v$ in the clique, add $2n$ new vertices connected only to $u,v$. The graph
has order $n+2n{n\choose 2} = n+n^2(n-1) = n^3-n^2+n\ge n^3/2$. Any minimal fvs
of this graph must contain at least $n-2$ vertices of the clique. As a result
its maximum size is at most $n-2+2n\le 3n$.  We have $\frac{\mmfvs(G)}{|V(G)|}
\le \frac{6n}{n^3} = O(\frac{1}{n^{2}})$ therefore
$\mmfvs(G)=O(|V(G)|^{1/3})$.  \end{proof}

Theorem \ref{thm:algpoly} also implies the existence of a cubic kernel of \textsc{Max Min FVS} when parameterized by the solution size $k$.
Recall that the reduction
rules do not change the solution size. 
We suppose that the reduced graph has $n$ vertices.
For a small constant $c$, if $n\ge c^3k^3$, then we can always produce a solution
of size at least $n^{1/3}/c=k$, and thus the answer is YES. 
Otherwise, we have a cubic kernel. 
\begin{corollary}\label{cor:kernel} \textsc{Max Min FVS} admits a cubic kernel when parameterized by the solution size. \end{corollary}

Finally, we remark that a similar combinatorial point of view can be taken for
the related problem of \textsc{Max Min VC}, giving another
intuitive explanation for the difference in approximability between the two
problems.

\begin{remark}\label{rem:vc} Any graph $G=(V,E)$ without isolated vertices, has a
minimal vertex cover of size at least $\sqrt{|V|}$, and this is asymptotically
tight.  \end{remark}

\begin{proof} We will prove the statement under the
assumption that $G$ is connected. If not, we can treat each component
separately. If the components of $G$ have sizes $n_1,\ldots,n_k$, then we rely
on the fact that $\sum_{i=1}^k \sqrt{n_i} \ge \sqrt{\sum_{i=1}^kn_i}$ and that
the union of the minimal vertex covers of each component is a minimal vertex
cover of $G$.

If $G=(V,E)$ has a vertex $u$ of degree at least $\sqrt{n}$, then we begin with
the vertex cover $V\setminus\{u\}$ and remove vertices until it becomes
minimal. In the end, our solution contains a superset of $N(u)$, therefore we
have a minimal vertex cover of size at least $\sqrt{n}$ as promised. If, on the
other hand, $\Delta(G)<\sqrt{n}$, then any vertex cover of $G$ must have size
at least $\sqrt{n}$. Indeed, a vertex cover of size at most $\sqrt{n}-1$ can
cover at most $(\sqrt{n}-1)\sqrt{n} < n-1$ edges, but since $G$ is connected we
have $|E(G)|\ge n-1$. So, in this case, any minimal vertex cover has the
promised size.

To see that the bound given is tight, take a $K_n$ and attach $n$ leaves to
each of its vertices. This graph has $n^2+n$ vertices, but any minimal vertex
cover has size at most $2n$.  \end{proof}

\section{Sub-exponential Time Approximation}

In this section we give an approximation algorithm that generalizes our
$n^{2/3}$-approximation and is able to guarantee any desired performance, at
the cost of increased running time. On a high level, our initial approach again
constructs an arbitrary minimal fvs $S$ and if $S$ is clearly large enough,
returns it. However, things become more complicated from then on, as it is no
longer sufficient to consider vertices of $S$ individually or in pairs. We
therefore need several new ideas, one of which is given in the following lemma,
which states that we can find a constant factor approximation in time
exponential in the size of a given fvs. This will be useful as we will use the
assumption that $S$ is ``small'' and then cut it up into even smaller pieces to
allow us to use Lemma \ref{lem:smallfvs}.

\begin{lemma}\label{lem:smallfvs} Given a graph $G=(V,E)$ on $n$ vertices and a
feedback vertex set $S\subseteq V$ of size $k$, it is possible to produce a
minimal fvs $S'$ of $G$ of size $|S'|\ge \frac{\mmfvs(G)}{3}$ in time
$n^{O(k)}$.  \end{lemma}

\begin{proof} Before we begin, let us point out that
for $k= 1$, \textsc{Max Min FVS} can be solved optimally in time $O(n)$, using
standard arguments from parameterized complexity, namely the fact that in this
case $G$ has treewidth $2$, and invoking Courcelle's theorem, since the
properties ``$S$ is an fvs'' and ``$S$ is minimal'' are MSO-expressible
\cite{Cygan2015}.  Unfortunately, this type of argument is not good enough for
larger values of $k$, as the running time guaranteed by Courcelle's theorem
could depend super-exponentially on $k$. We could try to avoid this by
formulating a treewidth-based DP algorithm to obtain a better running time, but
we prefer to give a simpler more direct branching algorithm, since this is good
enough for Theorem \ref{thm:subexp}.

We will assume that $S$ is minimal (if not, we can remove vertices from it to
make it minimal and this only decreases the available running time of our
algorithm). As a result, we assume that $\mmfvs(G)\ge 3k$, as otherwise $S$ is
already a $3$-approximation.

Let $S_{OPT}$ be a maximum minimal fvs of $G$, and $F_{OPT}=V\setminus
S_{OPT}$. We formulate an algorithm that maintains two disjoint sets of
vertices $S_{SOL}, F_{SOL}$ which, intuitively, correspond to vertices we have
decided to place in the fvs or the induced forest, respectively. We will denote
$U:=V\setminus (S_{SOL}\cup F_{SOL})$ the set of undecided vertices. Our
algorithm will be non-deterministic, that is, it will sometimes ``guess'' some
vertices of $U$ that will be placed in $S_{SOL}$ or $F_{SOL}$. We will bound
the total number of guessing possibilities by $n^{O(k)}$, which will imply that
the algorithm can be made deterministic by trying all possibilities for every
guess and returning the best returned solution.

Throughout the algorithm, we will work to maintain the following invariants:

\begin{enumerate}

\item\label{it:1} $S_{SOL}\cup F_{SOL}$ is an fvs of $G$.

\item\label{it:2}$S_{SOL}\subseteq S_{OPT}$ and $F_{SOL}\subseteq F_{OPT}$.

\item\label{it:3} $G[F_{SOL}]$ is acyclic and has at most $2k$ components.

\item\label{it:4} All vertices of $S_{SOL}$ have at least two neighbors in
$F_{SOL}$.

\end{enumerate}

To begin, we guess a set $F'\subseteq S$ such that $G[F']$ is acyclic and set
$F_{SOL}= F'$ and $S_{SOL}=S\setminus F'$. Property \ref{it:1} is satisfied as
$F_{SOL}\cup S_{SOL}=S$. Property \ref{it:2} is satisfied for the guess
$F'=F_{OPT}\cap S$. If there exists $u\in S_{SOL}$ which does not satisfy
Property \ref{it:4}, we guess one or two vertices from $N(u)\cap U$ and place
them into $F_{SOL}$ so that $u$ has two neighbors in $F_{SOL}$. Since $u$ has a
private cycle in $G[F_{OPT}]$, if the vertices we guessed are the neighbors of
$u$ in that cycle, we maintain Property \ref{it:2}. We continue in this way
until Property \ref{it:4} is satisfied. We now observe that $F_{SOL}$ is
acyclic (as $F_{SOL}\subseteq F_{OPT}$), and that since we have added at most
two vertices for each vertex of $S_{SOL}$, it contains at most $2k$ vertices,
hence at most $2k$ components, so we satisfied Property \ref{it:3}. So far, the
total number of possible guesses is upper-bounded by $2^kn^{2k}$: $2^k$ for
guessing $F'$ and $n^{2k}$ for guessing at most two neighbors for each $u\in
S_{SOL}$.

We will now say that a ``connector'' is a path $P\subseteq F_{OPT}\setminus
F_{SOL}$, such that $G[F_{OPT}\cup P]$ has strictly fewer components that
$G[F_{SOL}]$. Our algorithm will now repeateadly guess if a connector exists,
and if it does it will guess the first and last vertex $u,v$ of $P$. Note that
$u,v\in U$ and if we guess $u,v$ correctly we can infer all of $P$, as $G[U]$
is acyclic, so there is at most one path from $u$ to $v$ in $G[U]$. We set
$F_{SOL}:=F_{SOL}\cup P$ and continue guessing, until we guess that no
connector exists. Observe that guessing the endpoints of a connector gives $n^2$
possibilities, and that adding a connector to $F_{SOL}$ decreases the number of
connected components of $F_{SOL}$, which can happen at most $2k$ times by
Property \ref{it:3}. So we have a total of $n^{O(k)}$ possible guesses and for
the correct guess Property \ref{it:2} is maintained.

We now consider every vertex of $u\in U$ that has at least two neighbors in
$F_{SOL}$ and place all such vertices in $S_{SOL}$. Properties \ref{it:1},
\ref{it:3}, and \ref{it:4} are trivially still satisfied. Furthermore, if our
guesses so far are correct, all such vertices $u$ belong in $S_{OPT}$, as they
either already have a private cycle in $F_{OPT}$, or if they have neighbors in
distinct components of $F_{SOL}$, they would function as connectors in
$F_{OPT}$ (and we assume we have correctly guessed that no more connectors
exist).

We are now in a situation where every vertex of $U$ has at most one neighbor in
$F_{SOL}$. We construct a new graph $H$ by deleting from $G$ all of $S_{SOL}$
and replacing $F_{SOL}$ by a single vertex $f$ that is connected to
$N(F_{SOL})$. Note that $H$ is a simple graph (it has no parallel edges) with
an fvs of size $1$ (as $H-w$ is acyclic). We therefore use the aforementioned
algorithm implied by Courcelle's theorem to produce a maximum minimal fvs of
$H$ which, without loss of generality, does not contain $w$. Let $S^*\subseteq
U$ be this set. In $G$, we check if $S_{SOL}\cup S^*$ is an fvs. If it is we
delete vertices from it (if necessary) to make it redundant and return the
resulting set $S^{**}$, which is a minimal fvs.

To see that the resulting solution has the desired size we focus on the case
where all guesses were correct and therefore Properties \ref{it:1}-\ref{it:4}
were maintained throughout the execution of the algorithm. As mentioned, since
the total number of possibilities considered in $n^{O(k)}$, a deterministic
algorithm can simply try out all possible choices and return the best solution.

We first observe that $\mmfvs(H)\ge \mmfvs(G)-|S_{SOL}|$, where $S_{SOL}$ is
the set of vertices we deleted from $G$ to obtain $H$. Indeed,
$S_0:=S_{OPT}\setminus S_{SOL}$ is a minimal fvs of $H$. To see that $S_0$ is
an fvs, suppose that $H$ contains a cycle after deleting $S_0$. This cycle must
necessarily go through $w$. Let $P$ be the vertices of this cycle except $w$.
We have $P\subseteq U\setminus S_{OPT}$ therefore, $P\subseteq F_{OPT}$.
However, this means either that $P$ forms a cycle with a component of $F_{SOL}$
(which contradicts the acyclicity of $F_{OPT}$ by Property \ref{it:2}), or that
$P$ is a connector, which contradicts our guess that no other connector exists.
Therefore, $S_0$ must be an fvs of $H$. To see that it is minimal we note that
for all $u\in S_0$ there exists a private cycle in $G[U\cup F_{SOL}\cup\{u\}]$,
and this cycle is not destroyed by contracting the vertices of $F_{SOL}$ into~$w$.

We now have that $|S^*\cup S_{SOL}| \ge |S_{OPT}|$, because $|S^*|\ge
|S_{OPT}\setminus S_{SOL}|$. We argue that in the process of making $S^*$
minimal to obtain $S^{**}$ we delete at most $2k$ vertices. Indeed, every time
a vertex $u$ of $S_{SOL}$ is removed from $S^*\cup S_{SOL}$ as redundant, since
$u$ has at least two neighbors in $F_{SOL}$ by Property \ref{it:4}, the number
of components of $G[F_{SOL}]$ must decrease. Similarly, if we remove a vertex
$u\in S^*$ as redundant, we consider the private cycle of $u$ in $H\setminus
S^*$. All of the vertices of this cycle are present in $G$ after we delete
$S^*$, except $w$, therefore, this cycle forms a path between two distinct
components of $G[F_{SOL}]$. We conclude that, since removing a vertex from our
fvs decreases the number of connected components of $G[F_{SOL}]$, by Property
\ref{it:3} we have $|S^{**}|\ge |S_{OPT}|-2k$. But recall that we have assumed
that $k\le \frac{S_{OPT}}{3}$ (otherwise $S$ was already a sufficiently good
approximation), so we have $|S^{**}|\ge \frac{\mmfvs(G)}{3}$.  \end{proof}

\begin{theorem}\label{thm:subexp} There is an algorithm which, given an
$n$-vertex graph $G=(V,E)$ and a value $r$, produces an $r$-approximation for
\textsc{Max Min FVS} in $G$ in time $n^{O(n/r^{3/2})}$.  \end{theorem}

\begin{proof}

First, let us note that we may assume that $r$ is $\omega(1)$, because if $r$
is bounded by a constant, then we can solve the problem exactly in the given
time.  To ease presentation, we will give an algorithm with approximation ratio
$O(r)$. A ratio of exactly $r$ can be obtained by multiplying $r$ with an
appropriate (small) constant.

Our algorithm borrows several of the basic ideas from Theorem
\ref{thm:algpoly}, but requires some new ingredients (including Lemma
\ref{lem:smallfvs}). The first step is, again, to construct a minimal fvs $S$
in some arbitrary way, for example by setting $S=V$ and then removing vertices
from $S$ until it becomes minimal. If $|S|\ge n/r$ we are done, as we already
have an $r$-approximation, so we simply return $S$. From this point, this
algorithm departs from the algorithm of Theorem \ref{thm:algpoly}, because it
is no longer sufficient to compare the size of the returned solution with a
function of $n$ (we need to compare it to the actual optimal in order to obtain
a ratio of $r$), and because we need to partition $S$ into non-trivial parts
that contain more than one vertex. The algorithm proceeds as follows:

Let $k=\lceil \sqrt{r}\  \rceil$ and partition $S$ into $k$ parts of (almost)
equal size $S_1,\ldots,S_k$. Our algorithm proceeds as follows: for each
$i,j\in\{1,\ldots,k\}$ (not necessarily distinct) consider the graph $G_{i,j}$
obtained by deleting all vertices of $S\setminus (S_i\cup S_j)$. Compute, using
Lemma \ref{lem:smallfvs} a  solution for $G_{i,j}$, taking into account that
$S_i\cup S_j$ is a feedback vertex set of this graph. Output the largest of the
solutions found, using Lemma \ref{lem:ur1} to transform them into solutions of
$G$ (or output $S$ if it is larger than all solutions).

The algorithm clearly runs in the promised time: $|S_i\cup S_j|\le
\frac{2n}{rk}$, so the algorithm of Lemma \ref{lem:smallfvs} takes time
$n^{O(n/r^{3/2})}$ and is executed a polynomial number of times.

Let us now analyze the approximation ratio of the produced solution. Let
$S_{OPT}$ be an optimal solution and let $F:=V\setminus S$ and
$F_{OPT}=V\setminus S_{OPT}$ be the induced forests corresponding to $S$ and to
the optimal solution. We would like to argue that one of the considered
subproblems contains at least a $\frac{1}{r}$ fraction of $S_{OPT}$ and that
most (though not all) of these vertices form part of a minimal fvs of that
subgraph.

To be more precise, we will define the notion of ``type`` for each $u\in
S_{OPT}\cap F$. For each such $u$ there must exist a cycle in the graph
$G[F_{OPT}\cup\{u\}]$ (if not, this would contradict the minimality of
$S_{OPT}$). Call this cycle $c(u)$ (select one such cycle arbitrarily if
several exist). The cycle $c(u)$ must intersect $S$, as $S$ is an fvs. Let $v$
be the vertex of $c(u)\cap S$ closest to $u$ on the cycle.  Let $v'$ be the
vertex of $c(u)\cap S$ that is closest to $u$ if we traverse the cycle in the
opposite direction (note that $v,v'$ are not necessarily distinct). Suppose
that $v\in S_i, v'\in S_j$ and without loss of generality $i\le j$. We then say
that $u\in S_{OPT}\cap F$ has type $(i,j)$.  In this way, we define a type for
each $u\in S_{OPT}\cap F$. Note that according to our definition, all internal
vertices of the path in $c(u)$ from $u$ to $v$ (and also from $u$ to $v'$)
belong in $F_{OPT}\cap F$.

According to the definition of the previous paragraph, there are $k(k+1)/2\le
r$ possible types of vertices in $S_{OPT}\cap F$.  Therefore, there must be a
type $(i,j)$ such that at least $\frac{|S_{OPT}\cap F|}{r}$ vertices have this
type.  We now concentrate on the graph $G_{i,j}$, for the type $(i,j)$ which
satisfies this condition. Our algorithm constructed $G_{i,j}$ by deleting all
of $S$ except $S_i\cup S_j$.  We would like to claim that this graph has a
minimal feedback vertex set of size comparable to $\frac{|S_{OPT}\cap F|}{r}$.

For the sake of the analysis, construct a minimal feedback vertex set $S^*$ of
$G_{i,j}$ as follows: we begin with the fvs $S^*=S_{OPT}\cap(F\cup S_i\cup
S_j)$ and the corresponding induced forest $F^*=F_{OPT}\cap(F\cup S_i\cup
S_j)$. The set $S^*$ is a feedback vertex set as it contains all vertices of
$S_{OPT}$ found in $G_{i,j}$ and $S_{OPT}$ is a feasible feedback vertex set of
all of $G$. We then make $S^*$ minimal by arbitrarily removing redundant
vertices. Call the resulting set $S^{**}\subseteq S^*$ and the corresponding
induced forest $F^{**}\supseteq F^*$.

Our main claim now is that the number of vertices of $S^*\cap F$ of type
$(i,j)$ which were ``lost'' in the process of making $S^*$ minimal, is
upper-bounded by $|S_i\cup S_j|$. Formally, we claim that 
$|\{u\in(S^*\cap F)\setminus S^{**} \ |\ u \textrm{ has type } (i,j) \}|\le
|S_i\cup S_j|$.  Indeed, consider such a vertex $u\in (S^*\cap F)\setminus
S^{**}$ of type $(i,j)$, let $c(u)$ be the cycle that defines its type and
$v,v'$ the vertices of $S_i\cup S_j$ which are closest to $u$ on the cycle in
either direction. All vertices of $c(u)$ in the paths from $u$ to $v$ and from
$u$ to $v'$ belong to $F_{OPT}\cap F$, therefore also to $F^*$.  If $u$ was
removed as redundant, this means that $v,v'$ must have been in distinct
connected components at the moment $u$ was removed from the feedback vertex set
(and also that $v,v'$ are distinct). However, the addition of $u$ to the
induced forest creates a path from $v$ to $v'$ in the induced forest and hence
decreases the number of connected components (that is, trees in the induced
forest) containing vertices of $S_i\cup S_j$. The number of such connected
components cannot decrease more than $|S_i\cup S_j|$ times, therefore, during
the process of making $S^*$ minimal we may have removed at most $|S_i\cup S_j|$
vertices of type $(i,j)$ from $S^*\cap F$.

Using the above analysis and the assumption that $S^*$ contains at least
$\frac{|S_{OPT}\cap F|}{r}$ vertices of type $(i,j)$, we conclude that
$\mmfvs(G_{i,j})\ge |S^{**}|\ge \frac{|S_{OPT}\cap F|}{r} - |S_i\cup S_j|$. We
now note that if $|S_{OPT}\cap S|\ge \frac{|S_{OPT}|}{r}$, then $S$ is already
an $r$-approximation, so it is safe to assume $|S_{OPT}\cap F|\ge
\frac{(r-1)|S_{OPT}|}{r}$. Furthermore, $|S_i\cup S_j| \le
\frac{2|S|}{\sqrt{r}} \le \frac{2|S_{OPT}|}{r\sqrt{r}}$, where again we are
assuming that $S$ is not already an $r$-approximation. Putting things together
we get $\mmfvs(G_{i,j})\ge \frac{(r-1)|S_{OPT}|}{r^2} -
\frac{2|S_{OPT}|}{r\sqrt{r}} \ge \frac{|S_{OPT}|}{2r}$, for sufficiently large
$r$. Hence, since the algorithm will return a solution that is at least as
large as $\frac{\mmfvs(G_{i,j})}{3}$, we obtain an $O(r)$-approximation.
\end{proof}

\section{Hardness of Approximation and NP-hardness}\label{sec:inapprox}

In this section we establish lower bound results showing that the approximation
algorithms given in Theorems \ref{thm:algpoly} and \ref{thm:subexp} are
essentially optimal, under standard complexity assumptions.

\subsection{Hardness of Approximation in Polynomial Time}

We begin by showing that the best approximation ratio achievable in polynomial
time is indeed (essentially) $n^{2/3}$. For this, we rely on the celebrated
result of H{\aa}stad on the hardness of approximating \textsc{Max Independent
Set}, which was later derandomized by Zuckerman, cited below.

\begin{theorem}\cite{Hastad99,Zuckerman07} For any $\epsilon>0$, there is no
polynomial time algorithm which approximates \textsc{Max Independent Set} with
a ratio of $n^{1-\epsilon}$, unless $\text{P}=\text{NP}$. \end{theorem}

Starting from this result, we present a reduction to \textsc{Max Min FVS}.

\begin{theorem}\label{thm:Inapprox} For any $\epsilon>0$, \textsc{Max Min 
FVS} is inapproximable within a factor of $n^{2/3-\epsilon}$ unless $\text{P}=\text{NP}$.
\end{theorem} 

\begin{proof} We give a gap-preserving reduction from \textsc{Max Independent
Set}, which cannot be approximated within a factor of $n^{1-\epsilon}$, unless
$\text{P}=\text{NP}$.  We are given a graph $G=(V,E)$ on $n$ vertices as an instance of
\textsc{Max Independent Set}. Recall that $\alpha(G)$ denotes the size of the
maximum independent set of $G$.

We transform $G$ into an instance of \textsc{Max Min FVS} as follows: For
every pair of $u,v\in V$, we add $n$ vertices such that they are adjacent only
to $u$ and $v$.  We denote by $I_{uv}$ the set of such vertices.  Then $I_{uv}$
is an independent set. Let $G'=(V',E')$ be the constructed graph. 

We now make the following two claims:

\begin{claim} $\mmfvs(G')\ge (n-1){\alpha(G)\choose 2}$ \end{claim} 

\begin{proof} We construct a minimal fvs of $G'$ as follows: let $C$ be a
minimum vertex cover of $G$. Then we begin with the set that contains $C$ and
the union of all $I_{uv}$ (which is clearly an fvs) and remove vertices from it
until it becomes minimal. Let $S$ be the final minimal fvs. We observe that for
all $u,v\in V\setminus C$, $S$ contains at least $n-1$ of the vertices of
$I_{uv}$. Since $C$ is a minimum vertex cover of $G$, there are
$\alpha(G)\choose 2$ pairs $u,v\in V\setminus C$. \end{proof}

\begin{claim} $\mmfvs(G')\le n{2\alpha(G)\choose 2}+n$ \end{claim}

\begin{proof} 

Let $S$ be a minimal fvs of $G'$ and $F$ be the corresponding forest. It
suffices to show that $|S\setminus V|\le n{2\alpha(G)\choose 2}$, since $|S\cap
V|\le n$. Consider now a set $I_{uv}$. If $u\in S$ or $v\in S$, then
$I_{uv}\cap S=\emptyset$, because all vertices of $I_{uv}$ have at most one
neighbor in $F$, and are therefore redundant. So, $I_{uv}$ contains (at most
$n$) vertices of $S$ only if $u,v\in F$. However, $|F\cap V| \le 2 \alpha(G)$,
because $F$ is bipartite, so $F\cap V$ induces two independent sets, both of
which must be at most equal to the maximum independent set of $G$. So the
number of pairs $u,v\in F\cap V$ is at most $2\alpha(G)\choose 2$ and since
each corresponding $I_{uv}$ has size $n$, we get the promised bound.
\end{proof}

The two claims together imply that there exist constants $c_1,c_2$ such that
(for sufficiently large $n$) we have $c_1n(\alpha(G))^2 \le \mmfvs(G') \le
c_2n(\alpha(G))^2$. That is, $\mmfvs(G') = \Theta(n(\alpha(G))^2)$.

Suppose now that there exists a polynomial-time approximation algorithm which,
given a graph $G'$, produces a minimal fvs $S$ with the property
$\frac{\mmfvs(G')}{r}\le|S|\le \mmfvs(G')$, that is, there exists an
$r$-approximation for \textsc{Max Min FVS}. Running this algorithm on the
instance we constructed, we obtain that $\frac{c_1n(\alpha(G))^2}{r} \le |S|
\le c_2n(\alpha(G))^2$. Therefore,   $\frac{\alpha(G)}{\sqrt{rc_2/c_1}} \le
\sqrt{\frac{|S|}{c_2n}} \le \alpha(G)$. As a result, we obtain an $O(\sqrt{r})$
approximation for the value of $\alpha(G)$. We therefore conclude that, unless
$\text{P}=\text{NP}$, any such algorithm must have $\sqrt{r}>n^{1-\epsilon}$, for any
$\epsilon>0$, hence, $r>n^{2-\epsilon}$, for any $\epsilon>0$. Since the graph
$G'$ has $N=\Theta(n^3)$ vertices, we get that no approximation algorithm can
achieve a ratio of $N^{2/3-\epsilon}$. \end{proof}

We notice that in the construction of the previous theorem, the maximum degree
of the graph is approximately equal to the approximation gap.  Thus, the
following corollary also holds.

\begin{corollary}\label{cor:inapprox_deg} For any positive constant $\epsilon$,
\textsc{Max Min FVS} is inapproximable within a factor of
$\Delta^{1-\epsilon}$ unless $\text{P}=\text{NP}$.  \end{corollary}

\subsection{Hardness of Approximation in Sub-Exponential Time}

In this section we extend Theorem \ref{thm:Inapprox} to the realm of
sub-exponential time algorithms. We recall the following result of Chalermsook
et al.

\begin{theorem}\label{thm:chal}\cite{ChalermsookLN13} For any $\epsilon>0$
and any sufficiently large $r$, if there exists an $r$-approximation algorithm
for \textsc{Max Independent Set} running in $2^{(n/r)^{1-\epsilon}}$, then the
randomized ETH is false.  \end{theorem}

We remark that Theorem \ref{thm:chal}, which gives an almost tight running time
lower bound for \textsc{Max Independent Set}, has already been used as a
starting point to derive a similarly tight bound for the running time of any
sub-exponential time approximation for \textsc{Max Min VC}. Here,
we modify the proof of Theorem \ref{thm:Inapprox} to obtain a similarly tight
result for \textsc{Max Min FVS}. Nevertheless, the reduction for
\textsc{Max Min FVS} is significantly more challenging, because the ideas
used in Theorem \ref{thm:Inapprox} involve an inherent quadratic (in $n$)
blow-up of the size of the instance. As a result, in addition to executing an
appropriately modified version of the reduction of Theorem \ref{thm:Inapprox},
we are forced to add an extra ``sparsification'' step, and use a probabilistic
analysis with Chernoff bounds to argue that this step does not destroy the
inapproximability gap.

\begin{theorem}\label{thm:Inapprox2} For any $\epsilon>0$ and any sufficiently
large $r$, if there exists an $r$-approximation algorithm for \textsc{Max
Min FVS} running in $2^{(n/r^{3/2})^{1-\epsilon}}$, then the randomized ETH
is false.  \end{theorem}

\begin{proof} We recall some details about the
reduction used to prove Theorem \ref{thm:chal}. The reduction of
\cite{ChalermsookLN13} begins from a \textsc{3-SAT} instance $\phi$ on $n$
variables, and for any $\epsilon,r$, constructs a graph $G$ with
$n^{1+\epsilon}r^{1+\epsilon}$ vertices which (with high probability) satisfies
the following properties: if $\phi$ is satisfiable, then $\alpha(G)\ge
n^{1+\epsilon}r$; otherwise $\alpha(G)\le n^{1+\epsilon}r^{2\epsilon}$. Hence,
any approximation algorithm with ratio $r^{1-2\epsilon}$ for \textsc{Max
Independent Set} would be able to distinguish between the two cases (and solve
the initial \textsc{3-SAT} instance). If, furthermore, this algorithm runs in
$2^{(|V|/r)^{1-2\epsilon}}$, we get a sub-exponential algorithm for
\textsc{3-SAT}.

Suppose we are given $\epsilon,r$, and we want to prove the claimed lower bound
on the running time of any algorithm that $r$-approximates \textsc{Max Min 
FVS}. To ease presentation, we will assume that $r$ is the square of an integer
(this can be achieved without changing the value of $r$ by more than a small
constant).  We will also perform a reduction from \textsc{3-SAT} to show that
an algorithm that achieves this ratio too rapidly would give a sub-exponential
(randomized) algorithm for \textsc{3-SAT}.  We begin by executing the reduction
of \cite{ChalermsookLN13}, starting from a \textsc{3-SAT} instance $\phi$ on
$n$ variables, but adjusting their parameter $r$ appropriately so we obtain a
graph $G$ with the following properties (with high probability):

\begin{itemize}

\item $|V(G)| = n^{1+\epsilon}r^{1/2+\epsilon}$

\item If $\phi$ is satisfiable, then $\alpha(G)\ge n^{1+\epsilon}r^{1/2}$

\item If $\phi$ is not satisfiable, then $\alpha(G)\le
n^{1+\epsilon}r^{2\epsilon}$  

\end{itemize}

We now construct a graph $G'$ as follows: for each pair $u,v\in V(G)$, we
introduce an independent set $I_{uv}$ of size $\sqrt{r}$ connected to $u,v$. We
claim that $G'$ has the following properties (assuming $G$ has the properties
cited above):

\begin{itemize}

\item $|V(G')| = \Theta(n^{2+2\epsilon}r^{3/2+2\epsilon})$

\item If $\phi$ is satisfiable, then $\mmfvs(G') = \Omega(
n^{2+2\epsilon}r^{3/2})$

\item If $\phi$ is not satisfiable, then $\mmfvs(G') = O(
n^{2+2\epsilon}r^{1/2+4\epsilon})$

\end{itemize}

Before proceeding, let us establish the properties mentioned above. The size of
$|V(G')|$ is easy to bound, as for each of the $|V(G)|\choose 2$ pairs of
vertices of $G$ we have constructed an independent set of size $\sqrt{r}$. If
$\phi$ is satisfiable, we construct a minimal fvs of $G'$ by starting with a
minimum vertex cover of $G$ to which we add all vertices of all $I_{uv}$. We
then make this fvs minimal. We claim that for each $I_{uv}$ for which $u,v\in
V\setminus C$, our set will in the end contain all of $I_{uv}$, except maybe at
most one vertex. Furthermore, if one vertex of $I_{uv}$ is removed from the fvs
as redundant, this decreases the number of components of the induced forest
that contain vertices of $V$ (as $u,v$ are now in the same component). This
cannot happen more than $|V(G)|$ times. The number of $I_{uv}$ with $u,v\in
V\setminus C$ is ${\alpha(G)\choose 2} =\Omega(n^{2+2\epsilon}r)$. So,
$\mmfvs(G') = \Omega(n^{2+2\epsilon}r^{3/2} - |V(G)|)$. 

For the third property, take any minimal fvs $S$ of $G'$ and let $F$ be the
corresponding forest. We have $|F\cap V|\le 2\alpha(G)$, because $F$ is
bipartite. It is sufficient to bound $|S\setminus V|$ to obtain the bound (as
$|S\cap V|$ is already small enough). To do this, we note that in a set
$I_{uv}$ where $u,v$ are not both in $F$, we have $I_{uv}\cap S=\emptyset$, as
all vertices of $I_{uv}$ are redundant. So, the number of sets $I_{uv}$ which
contribute vertices to $S$ is at most ${|F\cap V|\choose 2} =
O(n^{2+2\epsilon}r^{4\epsilon})$. Each such set has size $\sqrt{r}$, giving the
claimed bound.

We have now constructed an instance where the gap between the values for
$\mmfvs(G')$, depending on whether $\phi$ is satisfiable, is almost $r$ (in
fact, it is $r^{1-4\epsilon}$, but we can make it equal to $r$ by adjusting the
parameters accordingly). The problem is that the order of the new graph depends
quadratically on $n$. This blow-up makes it impossible to obtain a running time
lower bound, as a fast approximation algorithm for \textsc{Max Min FVS}
(say with running time $2^{n/r^2}$) would not result in a sub-exponential
algorithm for \textsc{3-SAT}. We therefore need to ``sparsify'' our instance.

We construct a graph $G''$ by taking $G'$ and deleting every vertex of
$V(G')\setminus V(G)$ with probability $\frac{n-1}{n}$. That is, every vertex
of the independent sets $I_{uv}$ we added survives (independently) with
probability $1/n$. We now claim the following properties hold with high
probability:

\begin{itemize}

\item $|V(G'')| = \Theta(n^{1+2\epsilon}r^{3/2+2\epsilon})$

\item If $\phi$ is satisfiable, then $\mmfvs(G'') =
\Omega(n^{1+2\epsilon}r^{3/2})$

\item If $\phi$ is not satisfiable, then $\mmfvs(G'') =
O(n^{1+2\epsilon}r^{1/2+4\epsilon})$

\end{itemize}

Before we proceed, let us explain why if we establish that $G''$ satisfies
these properties, then we obtain the theorem. Indeed, suppose that for some
sufficiently large $r$ and $\epsilon>0$, there exists an approximation
algorithm for \textsc{Max Min FVS} with ratio $r^{1-5\epsilon}$ running in
time $2^{(N/r^{3/2})^{1-10\epsilon}}$ for graphs with $N$ vertices. The
algorithm has sufficiently small ratio to distinguish between the two cases in
our constructed graph $G''$, as the ratio between $\mmfvs(G'')$ when $\phi$ is
satisfiable or not is $\Omega(r^{1-4\epsilon})$ (and $r$ is sufficiently
large), so we can use the approximation algorithm to solve \textsc{3-SAT}.
Furthermore, to compute the running time we see that $N/r^{3/2} =
\Theta(n^{1+2\epsilon}r^{2\epsilon}) = O(n^{1+4\epsilon})$.  Therefore,
$(N/r^{3/2})^{1-10\epsilon} = o(n)$ and we get a sub-exponential time algorithm
for \textsc{3-SAT}. We conclude that for any sufficiently large $r$ and any
$\epsilon>0$, no algorithm achieves ratio $r^{1-5\epsilon}$ in time
$2^{(N/r^{3/2})^{1-10\epsilon}}$. By adjusting $r,\epsilon$ appropriately we
get the statement of the theorem.

Let us therefore try to establish that the three claimed properties all hold
with high probability. We will use the following standard Chernoff bound:
suppose $X=\sum_{i=1}^n X_i$ is the sum of $n$ independent random $0/1$
variables $X_i$ and that $E[X] = \sum_{i=1}^n E[X_i] = \mu$. Then, for all
$\delta\in (0,1)$ we have $Pr[|X-\mu| \ge \delta \mu] \le 2e^{-\mu\delta^2/3}$

The first property is easy to establish: we define a random variable $X_i$ for
each vertex of each $I_{uv}$ of $G'$. This variable takes value $1$ if the
corresponding vertex appears in $G''$ and $0$ otherwise. Let $X$ be the sum of
the $X_i$ variables, which corresponds to the number of such vertices appearing
in $G''$. Suppose that the number of vertices in sets $I_{uv}$ of $G'$ is
$cn^{2+2\epsilon}r^{3/2+2\epsilon}$, where $c$ is a constant. Then,
$E[X]=cn^{1+2\epsilon}r^{3/2+2\epsilon}$. Also, $Pr[|X-E[X]| \ge
\frac{E[X]}{2}] \le 2e^{-E[X]/12} = o(1)$. So with high probability, $|V(G'')|$
is of the promised magnitude.

The second property is also straightforward. This time we consider a maximum
minimal fvs $S$ of $G'$ of size $cn^{2+2\epsilon}r^{3/2}$. Again, we define an
indicator variable for each vertex of this set in sets $I_{uv}$. The expected
number of such vertices that survive in $G''$ is $cn^{1+2\epsilon}r^{3/2}$. As
in the previous paragraph, with high probability the actual number will be
close to this bound. We now need to argue that (almost) the same set is a
minimal fvs of $G''$. We start in $G''$ with (the surviving vertices of) $S$,
which is clearly an fvs of $G''$, and delete vertices until the set is minimal.
We claim that the size of the set will decrease by at most
$|V(G)|=n^{1+\epsilon}r^{1+\epsilon}$.  Indeed, if $S\cap I_{uv}\neq
\emptyset$, then $u,v\not\in S$. The two vertices $u,v$ are (deterministically)
included in $G''$ and start out in the corresponding induced forest in our
solution. If a vertex of $S\cap I_{uv}$ is deleted as redundant, placing that
vertex in the forest will put $u,v$ in the same component, reducing the number
of components of the forest with vertices from $|V(G)|$. This can happen at
most $|V(G)|$ times.  Since $|V(G)|<\frac{c}{10}(n^{1+2\epsilon}r^{3/2})$ (for
$n,r$ sufficiently large), deleting these redundant vertices will not change
the order of magnitude of the solution.

Finally, in order to establish the third property we need to consider every
possible minimal fvs of $G''$ and show that none of them end up being too
large. Consider a set $F\subseteq V(G)$ that induces a forest in $G$. Our goal
is to prove that any minimal fvs $S$ of $G''$ that satisfies $ V(G)\setminus
S=F$ has a probability of being ``too large'' (that is, violating our claimed
bound) much smaller than $2^{-|V(G)|}$. If we achieve this, then we can take a
union bound over all sets $F$ and conclude that with high probability no
minimal fvs of $G''$ is too large.

Suppose then that we have fixed an acyclic set $F\subseteq V(G)$. We have
$|F|\le 2\alpha(G) \le 2n^{1+\epsilon}r^{2\epsilon}$. Any minimal fvs with
$V(G)\setminus S= F$ can only contain vertices from a set $I_{uv}$ if $u,v\in
F$. The total number of such vertices in $G'$ is at most
$O(n^{2+2\epsilon}r^{1/2+4\epsilon})$. The expected number of such vertices
that survive in $G''$ is (for some constant $c$) at most
$\mu=cn^{1+2\epsilon}r^{1/2+4\epsilon}$. Now, using the Chernoff bound cited
above we have $Pr[|X-\mu|\ge\frac{\mu}{2}] \le 2e^{-\mu/12}$. We claim
$2e^{-\mu/12} = o(2^{-|V(G)|})$. Indeed, this follows because $|V(G)| =
n^{1+\epsilon}r^{1/2+\epsilon} = o(\mu)$. As a result, the probability that a
large minimal fvs exists for a fixed set $F\subseteq V(G)$ exists is low enough
that taking the union bound over all possible sets $F$ we have that with high
probability no minimal fvs exists with value higher than $3\mu/2$, which
establishes the third property. \end{proof}

\subsection{NP-hardness for $\Delta=6$}\label{sec:NP}

\begin{theorem}\label{thm:NP-h}
\textsc{Max Min FVS} is NP-hard on planar bipartite graphs with $\Delta=6$.
\end{theorem}

\begin{proof} We give a reduction from \textsc{Max Min VC},
which is NP-hard on planar bipartite graphs of maximum degree 3 \cite{ZZ1995}.
Note that the NP-hardness in \cite{ZZ1995} is stated for \textsc{Minimum
Independent Dominating Set}, but any independent dominating set is also a
maximal independent set (and vice-versa) and the complement of the minimum
maximal independent set of any graph is a maximum minimal vertex cover. Thus,
we also obtain NP-hardness for \textsc{Max Min VC} on the same instances.

We are given a graph $G=(V,E)$.  For each edge $e=(u,v)\in E$, we add a path of
length three from $u$ to $v$ going through two new vertices $e^{(1)},e^{(2)}$
(see Figure \ref{fig_core_K6}).  Note that $u,e^{(1)},e^{(2)},v$ form a cycle
of length 4.  Then we add two cycles of length 4,
$e^{(i)},c^{(i)}_{e1},c^{(i)}_{e2},c^{(i)}_{e3}$ and
$e^{(i)},c^{(i)}_{e4},c^{(i)}_{e5},c^{(i)}_{e6}$  for $i\in \{1,2\}$.  Let
$G'=(V',E')$ be the constructed graph.  Because  $\Delta(G)=3$, we have
$\Delta(G')=6$.  Moreover, since $G$ is planar and bipartite, $G'$ is also
planar and bipartite.  We will show that there is a minimal vertex cover of
size at least $k$ in $G$ if and only if there is a minimal feedback vertex set
of size at least $k+4|E|$ in $G'$.

\begin{figure}[tbp]
    \centering
    \includegraphics[width=8cm]{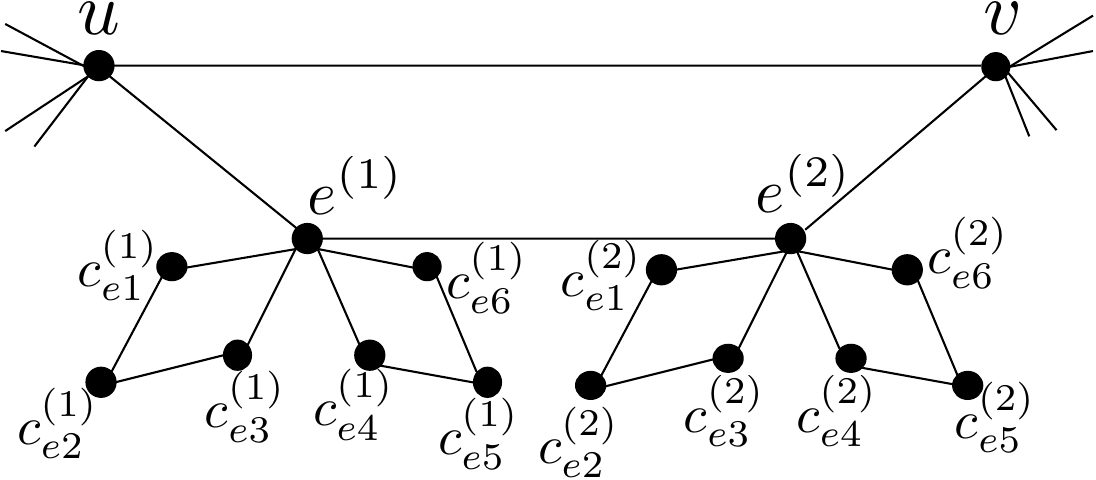}
    \caption{The edge gadget of $e=(u,v)$ in the constructed graph $G$.}
    \label{fig_core_K6}
\end{figure}

Given a minimal vertex cover $S$ of size at least $k$ in $G$, we construct the
set $S'=S\cup \bigcup_{e\in E}
\{c_{e1}^{(1)},c_{e4}^{(1)},c_{e1}^{(2)},c_{e4}^{(2)}\}$.  Then $|S'|\ge
k+4|E|$. Let us first argue that $S'$ is an fvs of $G'$. For each $e=(u,v)\in
E$ we have at least one of $u,v\in S$, without loss of generality let $u\in S$.
Now in $G'[V'\setminus S']$ the edges $(e^{(1)},e^{(2)})$ and $(e^{(2)},v)$ are
bridges and therefore cannot be part of any cycle. The remaining cycles going
through $e^{(1)},e^{(2)}$ are handled by
$\{c_{e1}^{(1)},c_{e4}^{(1)},c_{e1}^{(2)},c_{e4}^{(2)}\}$.  Furthermore, since
$G'[V\setminus S]$ is an independent set, it is also acyclic. To see that $S'$
is a minimal fvs, we remark that for each $c_{e1}^{(i)},c_{e4}^{(i)}$ contained
in $S'$ there is a private cycle in $G'[V'\setminus S']$. We also note that
since $S$ is a minimal vertex cover of $G$, for each $u\in S$, there exists
$v\not\in S$ with $e=(u,v)\in E$. This means that $u$ has the private cycle
formed by $\{u,v,e^{(1)},e^{(2)}\}$ in $G'[V'\setminus S']$.  Therefore, $S'$
is a minimal fvs.

Conversely, suppose we are given a minimal fvs $S'$ of $G'$ with $|S'|\ge
k+4|E|$.  We will edit $S'$ so that is contains only vertices in $V'\setminus
\bigcup_{e\in E}\{e^{(1)},e^{(2)}\}$, without decreasing its size. 

First, suppose $e^{(1)},e^{(2)}\in S'$, for some $e\in E$. We construct a new
minimal fvs $S''=S'\setminus\{e^{(2)}\}\cup\{c_{e1}^{(2)}, c_{e4}^{(2)}\}$
which is larger that $S'$, since by minimality we have $c_{ei}^{(2)}\not\in S'$
for $i\in\{1,\ldots,6\}$.  It is not hard to see that $S''$ is indeed an fvs,
as no cycle can go through $e^{(2)}$ in $G'[V'\setminus S'']$.  The two
vertices we added have a private cycle, while all vertices of $S'\cap S''$
retain their private cycles, so $S''$ is a minimal fvs. As a result in the
remainder we assume that $S'$ contains at most one of $\{e^{(1)}, e^{(2)}\}$
for all $e\in E$.

Suppose now that for some $e=(u,v)\in E$, we have $S'\cap\{u,v\}\neq\emptyset$
and $S'\cap\{e^{(1)},e^{(2)}\}\neq \emptyset$. Without loss of generality, let
$e^{(1)}\in S'$. We set
$S''=S'\setminus\{e^{(1)}\}\cup\{c_{e1}^{(1)},c_{e4}^{(1)}\}$ and claim that
$S''$ is a larger minimal fvs than $S$. Indeed, no cycle goes through $e^{(1)}$
in $G'[V'\setminus S'']$, the new vertices we added to $S'$ have private
cycles, and all vertices of $S'\cap S''$ retain their private cycles in
$G'[V'\setminus S'']$.  Therefore, we can now assume that if for some
$e=(u,v)\in E$ we have $S'\cap \{e^{(1)},e^{(2)}\}\neq \emptyset$ then
$u,v\not\in S'$.

For the remaining case, suppose that for some $e=(u,v)\in E$ we have
$u,v\not\in S'$ and (without loss of generality) $e^{(1)}\in S'$. We construct
the set $S''=S'\setminus\{e^{(1)}\}\cup\{c_{e1}^{(1)},c_{e4}^{(2)},u\}$. Note
that $|S''|\ge |S'|+2$. It is not hard to see that $S''$ is an fvs, since by
adding $c_{e1}^{(1)}, c_{e4}^{(1)}, v$ to our set we have hit all cycles
containing $e^{(1)}$ in $G'$. The problem now is that $S''$ is not necessarily
minimal. We greedily delete vertices from $S''$ to obtain a minimal fvs $S^*$.
We claim that in this process we cannot delete more than two vertices, that is
$|S^*\setminus S''|\le 2$.  To see this, we first note that
$c_{e1}^{(1)},c_{e4}^{(2)},u$ cannot be removed from $S''$ as they have private
cycles in $G[V'\setminus S'']$. Suppose now that $w_1\in S''\setminus S^*$ is
the first vertex we removed from $S''$, so $G'[(V'\setminus S'')\cup \{w_1\}]$
is acyclic.  This vertex must have had a private cycle in $G'[V'\setminus S']$,
which was necessarily going through $u$. Therefore, $G'[(V'\setminus S'')\cup
\{w_1\}]$ has a path connecting two neighbors of $u$ and this path does not
exist in $G'[(V'\setminus S'')]$. With a similar reasoning, removing another
vertex $w_2\in S''$ from the fvs will create a second path between neighbors of
$u$ in the induced forest. We conclude that this cannot happen a third time,
since $|N(u)|\le 3$, and if we create three paths between neighbors of $u$,
this will create a cycle. As a result, $|S^*|\ge |S'|$. We assume in the
remainder that $S'$ does not contain $e^{(1)},e^{(2)}$ for any $e\in E$.

Now, given a minimal fvs $S'$ of $G'$ with $|S'|\ge k+4|E|$ and $S'\cap
(\cup_{e\in E} \{e^{(1)},e^{(2)}\}) = \emptyset$ we set $S=S'\cap V$ and claim
that $S$ is a minimal vertex cover of $G$ with $|S|\ge k$. Indeed $S$ is a
vertex cover, as for each $e=(u,v)\in E$, if $u,v\not\in S'$ then we would get
the cycle formed by $\{u,v,e^{(1)},e^{(2)}\}$. To see that $S$ is minimal,
suppose $N_G[u]\subseteq S'$. We claim that in that case $u$ has no private
cycle in $G'[V'\setminus S']$ (this can be seen by deleting all bridges in
$G'[V'\setminus S']$, which leaves $u$ isolated). This contradicts the
minimality of $S'$ as an fvs of $G'$. Finally, we argue that $|S'\setminus
V|\le 4|E|$, which gives the desired bound on $|S|$. Consider an $e=(u,v)\in
E$. $S'$ cannot contain more than one vertex among
$c_{e1}^{(1)},c_{e2}^{(1)},c_{e3}^{(1)}$, since any of these vertices hits the
cycle that goes through the others. With similar reasoning for the three other
length-four cycles we conclude that $S'$ contains at most $4$ vertices for each
edge $e\in E$.  \end{proof}

\section{Conclusions}

We have essentially settled the approximability of \textsc{Max Min FVS} for
polynomial and sub-exponential time, up to sub-polynomial factors in the
exponent of the running time. It would be interesting to see if the running
time of our sub-exponential approximation algorithm can be improved by
poly-logarithmic factors in the exponent, as in \cite{BansalCLNN19}. In
particular, improving the running time to $2^{O(n/r^{3/2})}$ seems feasible,
but would likely require a version of Lemma \ref{lem:smallfvs} which uses more
sophisticated techniques, such as Cut\&Count
\cite{BodlaenderCKN15,Cygan2015,CyganNPPRW11}. 
For the parameterized complexity perspective, we gave a cubic kernel when parameterized by solution size.
A natural direction of future work is the deep analysis of  parameterized complexity of \textsc{Max Min FVS}.
Finally, we showed that  \textsc{Max Min FVS} is NP-hard even on graphs of maximum degree 6.
An interesting open problem is the complexity on graphs of maximum degree 3, where \textsc{Min FVS} can be solved in polynomial time \cite{UENO1988}.

Another problem of similar spirit which deserves to be studied is \textsc{Max
Min OCT}, where an odd cycle transversal (OCT) is a set of vertices whose
removal makes the graph bipartite. This problem could also potentially be
``between'' \textsc{Max Min VC} and \textsc{UDS}, but obtaining a
$n^{1-\epsilon}$ approximation for it seems much more challenging than for
\textsc{Max Min FVS}.

\section*{Acknowledgments}

This work is partially supported by PRC CNRS JSPS project PARAGA and by JSPS KAKENHI Grant Number JP19K21537.

\bibliographystyle{plain}
\bibliography{ref}

\begin{thebibliography}{10}

\bibitem{AbouEishaHLMRZ18}
Hassan AbouEisha, Shahid Hussain, Vadim~V. Lozin, J{\'{e}}r{\^{o}}me Monnot,
  Bernard Ries, and Viktor Zamaraev.
\newblock Upper domination: Towards a dichotomy through boundary properties.
\newblock {\em Algorithmica}, 80(10):2799--2817, 2018.

\bibitem{AboulkerB0S20}
Pierre Aboulker, {\'{E}}douard Bonnet, Eun~Jung Kim, and Florian Sikora.
\newblock Grundy coloring {\&} friends, half-graphs, bicliques.
\newblock In {\em {STACS}}, volume 154 of {\em LIPIcs}, pages 58:1--58:18.
  Schloss Dagstuhl - Leibniz-Zentrum f{\"{u}}r Informatik, 2020.

\bibitem{ArkinBMS03}
Esther~M. Arkin, Michael~A. Bender, Joseph S.~B. Mitchell, and Steven Skiena.
\newblock The lazy bureaucrat scheduling problem.
\newblock {\em Inf. Comput.}, 184(1):129--146, 2003.

\bibitem{BansalCLNN19}
Nikhil Bansal, Parinya Chalermsook, Bundit Laekhanukit, Danupon Nanongkai, and
  Jesper Nederlof.
\newblock New tools and connections for exponential-time approximation.
\newblock {\em Algorithmica}, 81(10):3993--4009, 2019.

\bibitem{Bazgan2018}
Cristina Bazgan, Ljiljana Brankovic, Katrin Casel, Henning Fernau, Klaus
  Jansen, Kim{-}Manuel Klein, Michael Lampis, Mathieu Liedloff,
  J{\'{e}}r{\^{o}}me Monnot, and Vangelis~Th. Paschos.
\newblock The many facets of upper domination.
\newblock {\em Theoretical Computer Science}, 717:2--25, 2018.

\bibitem{BelmonteKLMO20}
R{\'e}my Belmonte, Eun~Jung Kim, Michael Lampis, Valia Mitsou, and Yota Otachi.
\newblock {Grundy Distinguishes Treewidth from Pathwidth}.
\newblock In Fabrizio Grandoni, Grzegorz Herman, and Peter Sanders, editors,
  {\em 28th Annual European Symposium on Algorithms (ESA 2020)}, volume 173 of
  {\em Leibniz International Proceedings in Informatics (LIPIcs)}, pages
  14:1--14:19, Dagstuhl, Germany, 2020. Schloss Dagstuhl--Leibniz-Zentrum
  f{\"u}r Informatik.

\bibitem{BodlaenderCKN15}
Hans~L. Bodlaender, Marek Cygan, Stefan Kratsch, and Jesper Nederlof.
\newblock Deterministic single exponential time algorithms for connectivity
  problems parameterized by treewidth.
\newblock {\em Inf. Comput.}, 243:86--111, 2015.

\bibitem{Bonnet2018}
{\'{E}}douard Bonnet, Michael Lampis, and Vangelis~Th. Paschos.
\newblock Time-approximation trade-offs for inapproximable problems.
\newblock {\em Journal of Computer and System Sciences}, 92:171 -- 180, 2018.

\bibitem{Boria2015}
Nicolas Boria, Federico~Della Croce, and Vangelis~Th. Paschos.
\newblock On the max min vertex cover problem.
\newblock {\em Discrete Applied Mathematics}, 196:62--71, 2015.

\bibitem{BourgeoisEP09}
Nicolas Bourgeois, Bruno Escoffier, and Vangelis~Th. Paschos.
\newblock Approximation of min coloring by moderately exponential algorithms.
\newblock {\em Inf. Process. Lett.}, 109(16):950--954, 2009.

\bibitem{BoyaciM17}
Arman Boyaci and J{\'{e}}r{\^{o}}me Monnot.
\newblock Weighted upper domination number.
\newblock {\em Electron. Notes Discret. Math.}, 62:171--176, 2017.

\bibitem{ChalermsookLN13}
Parinya Chalermsook, Bundit Laekhanukit, and Danupon Nanongkai.
\newblock Independent set, induced matching, and pricing: Connections and tight
  (subexponential time) approximation hardnesses.
\newblock In {\em {FOCS}}, pages 370--379. {IEEE} Computer Society, 2013.

\bibitem{ChestonFHJ90}
Grant~A. Cheston, Gerd Fricke, Stephen~T. Hedetniemi, and David~Pokrass Jacobs.
\newblock On the computational complexity of upper fractional domination.
\newblock {\em Discret. Appl. Math.}, 27(3):195--207, 1990.

\bibitem{CourcelleMR00}
Bruno Courcelle, Johann~A. Makowsky, and Udi Rotics.
\newblock Linear time solvable optimization problems on graphs of bounded
  clique-width.
\newblock {\em Theory Comput. Syst.}, 33(2):125--150, 2000.

\bibitem{Cygan2015}
Marek Cygan, Fedor~V. Fomin, Lukasz Kowalik, Daniel Lokshtanov, D{\'{a}}niel
  Marx, Marcin Pilipczuk, Michal Pilipczuk, and Saket Saurabh.
\newblock {\em Parameterized Algorithms}.
\newblock Springer International Publishing, 2015.

\bibitem{CyganKW09}
Marek Cygan, Lukasz Kowalik, and Mateusz Wykurz.
\newblock Exponential-time approximation of weighted set cover.
\newblock {\em Inf. Process. Lett.}, 109(16):957--961, 2009.

\bibitem{CyganNPPRW11}
Marek Cygan, Jesper Nederlof, Marcin Pilipczuk, Michal Pilipczuk, Johan M.~M.
  van Rooij, and Jakub~Onufry Wojtaszczyk.
\newblock Solving connectivity problems parameterized by treewidth in single
  exponential time.
\newblock In {\em {FOCS}}, pages 150--159. {IEEE} Computer Society, 2011.

\bibitem{CyganP10}
Marek Cygan and Marcin Pilipczuk.
\newblock Exact and approximate bandwidth.
\newblock {\em Theor. Comput. Sci.}, 411(40-42):3701--3713, 2010.

\bibitem{Demange1999}
Marc Demange.
\newblock A note on the approximation of a minimum-weight maximal independent
  set.
\newblock {\em Computational Optimization and Applications}, 14(1):157--169,
  1999.

\bibitem{Confpaper}
Louis Dublois, Tesshu Hanaka, Mehdi~Khosravian Ghadikolaei, Michael Lampis, and
  Nikolaos Melissinos.
\newblock {(In)approximability of Maximum Minimal FVS}.
\newblock In Yixin Cao, Siu-Wing Cheng, and Minming Li, editors, {\em 31st
  International Symposium on Algorithms and Computation (ISAAC 2020)}, volume
  181 of {\em Leibniz International Proceedings in Informatics (LIPIcs)}, pages
  3:1--3:14, Dagstuhl, Germany, 2020. Schloss Dagstuhl--Leibniz-Zentrum f{\"u}r
  Informatik.

\bibitem{EscoffierPT14}
Bruno Escoffier, Vangelis~Th. Paschos, and Emeric Tourniaire.
\newblock Approximating {MAX} {SAT} by moderately exponential and parameterized
  algorithms.
\newblock {\em Theor. Comput. Sci.}, 560:147--157, 2014.

\bibitem{EtoHKK2019}
Hiroshi Eto, Tesshu Hanaka, Yasuaki Kobayashi, and Yusuke Kobayashi.
\newblock {Parameterized Algorithms for Maximum Cut with Connectivity
  Constraints}.
\newblock In Bart M.~P. Jansen and Jan~Arne Telle, editors, {\em 14th
  International Symposium on Parameterized and Exact Computation (IPEC 2019)},
  volume 148 of {\em Leibniz International Proceedings in Informatics
  (LIPIcs)}, pages 13:1--13:15, Dagstuhl, Germany, 2019. Schloss
  Dagstuhl--Leibniz-Zentrum fuer Informatik.

\bibitem{FotakisLP16}
Dimitris Fotakis, Michael Lampis, and Vangelis~Th. Paschos.
\newblock Sub-exponential approximation schemes for csps: From dense to almost
  sparse.
\newblock In {\em {STACS}}, volume~47 of {\em LIPIcs}, pages 37:1--37:14.
  Schloss Dagstuhl - Leibniz-Zentrum f{\"{u}}r Informatik, 2016.

\bibitem{Furini2017}
Fabio Furini, Ivana Ljubic, and Markus Sinnl.
\newblock An effective dynamic programming algorithm for the minimum-cost
  maximal knapsack packing problem.
\newblock {\em European Journal of Operational Research}, 262(2):438--448,
  2017.

\bibitem{GourvesMP13}
Laurent Gourv{\`{e}}s, J{\'{e}}r{\^{o}}me Monnot, and Aris Pagourtzis.
\newblock The lazy bureaucrat problem with common arrivals and deadlines:
  Approximation and mechanism design.
\newblock In {\em {FCT}}, volume 8070 of {\em Lecture Notes in Computer
  Science}, pages 171--182. Springer, 2013.

\bibitem{Hanaka2019}
Tesshu Hanaka, Hans~L. Bodlaender, Tom~C. van~der Zanden, and Hirotaka Ono.
\newblock On the maximum weight minimal separator.
\newblock {\em Theoretical Computer Science}, 796:294 -- 308, 2019.

\bibitem{DBLP:conf/ttcs/HarutyunyanGMMP20}
Ararat Harutyunyan, Mehdi~Khosravian Ghadikolaei, Nikolaos Melissinos,
  J{\'{e}}r{\^{o}}me Monnot, and Aris Pagourtzis.
\newblock On the complexity of the upper r-tolerant edge cover problem.
\newblock In Lu{\'{\i}}s~Soares Barbosa and Mohammad~Ali Abam, editors, {\em
  Topics in Theoretical Computer Science - Third {IFIP} {WG} 1.8 International
  Conference, {TTCS} 2020, Tehran, Iran, July 1-2, 2020, Proceedings}, volume
  12281 of {\em Lecture Notes in Computer Science}, pages 32--47. Springer,
  2020.

\bibitem{Hastad99}
Johan H{\aa}stad.
\newblock Clique is hard to approximate within $n^{1-\epsilon}$.
\newblock {\em Acta Math}, 182:105--142, 1999.

\bibitem{HenningP20}
Michael~A. Henning and Dinabandhu Pradhan.
\newblock Algorithmic aspects of upper paired-domination in graphs.
\newblock {\em Theor. Comput. Sci.}, 804:98--114, 2020.

\bibitem{DBLP:journals/combinatorics/HenningY18}
Michael~A. Henning and Anders Yeo.
\newblock On upper transversals in 3-uniform hypergraphs.
\newblock {\em Electron. J. Comb.}, 25(4):P4.27, 2018.

\bibitem{DBLP:journals/ejc/HenningY19}
Michael~A. Henning and Anders Yeo.
\newblock Upper transversals in hypergraphs.
\newblock {\em Eur. J. Comb.}, 78:1--12, 2019.

\bibitem{DBLP:journals/jco/HenningY20}
Michael~A. Henning and Anders Yeo.
\newblock Bounds on upper transversals in hypergraphs.
\newblock {\em J. Comb. Optim.}, 39(1):77--89, 2020.

\bibitem{IwaideN16}
Ken Iwaide and Hiroshi Nagamochi.
\newblock An improved algorithm for parameterized edge dominating set problem.
\newblock {\em J. Graph Algorithms Appl.}, 20(1):23--58, 2016.

\bibitem{JacobsonP90}
Michael~S. Jacobson and Kenneth Peters.
\newblock Chordal graphs and upper irredundance, upper domination and
  independence.
\newblock {\em Discret. Math.}, 86(1-3):59--69, 1990.

\bibitem{KatsikarelisLP19}
Ioannis Katsikarelis, Michael Lampis, and Vangelis~Th. Paschos.
\newblock Improved (in-)approximability bounds for d-scattered set.
\newblock In {\em {WAOA}}, volume 11926 of {\em Lecture Notes in Computer
  Science}, pages 202--216. Springer, 2019.

\bibitem{KhoshkhahGMS20}
Kaveh Khoshkhah, Mehdi~Khosravian Ghadikolaei, J{\'{e}}r{\^{o}}me Monnot, and
  Florian Sikora.
\newblock Weighted upper edge cover: Complexity and approximability.
\newblock {\em J. Graph Algorithms Appl.}, 24(2):65--88, 2020.

\bibitem{MishraS01}
Sounaka Mishra and Kripasindhu Sikdar.
\newblock On the hardness of approximating some {NP}-optimization problems
  related to minimum linear ordering problem.
\newblock {\em {RAIRO} Theor. Informatics Appl.}, 35(3):287--309, 2001.

\bibitem{DBLP:conf/wg/MisraRRS12}
Pranabendu Misra, Venkatesh Raman, M.~S. Ramanujan, and Saket Saurabh.
\newblock Parameterized algorithms for even cycle transversal.
\newblock In Martin~Charles Golumbic, Michal Stern, Avivit Levy, and Gila
  Morgenstern, editors, {\em Graph-Theoretic Concepts in Computer Science -
  38th International Workshop, {WG} 2012, Jerusalem, Israel, June 26-28, 2012,
  Revised Selcted Papers}, volume 7551 of {\em Lecture Notes in Computer
  Science}, pages 172--183. Springer, 2012.

\bibitem{UENO1988}
Shuichi Ueno, Yoji Kajitani, and Shin'ya Gotoh.
\newblock On the nonseparating independent set problem and feedback set problem
  for graphs with no vertex degree exceeding three.
\newblock {\em Discrete Mathematics}, 72(1):355 -- 360, 1988.

\bibitem{Zehavi2017}
Meirav Zehavi.
\newblock Maximum minimal vertex cover parameterized by vertex cover.
\newblock {\em SIAM Journal on Discrete Mathematics}, 31(4):2440--2456, 2017.

\bibitem{Zuckerman07}
David Zuckerman.
\newblock Linear degree extractors and the inapproximability of max clique and
  chromatic number.
\newblock {\em Theory of Computing}, 3(1):103--128, 2007.

\bibitem{ZZ1995}
Igor~E. Zvervich and Vadim~E. Zverovich.
\newblock An induced subgraph characterization of domination perfect graphs.
\newblock {\em Journal of Graph Theory}, 20(3):375--395, 1995.

\end{thebibliography}

\end{document}